\documentclass[oneside]{article}
\usepackage{mathrsfs,amssymb,amsmath,amsthm,amsfonts, cite}
\usepackage{graphicx, float,epstopdf}
\usepackage{enumitem}
\usepackage{diagbox}
\usepackage{makecell}
\usepackage{booktabs,multirow}
\usepackage{algorithm}
\usepackage{algpseudocode}
\usepackage{subfigure}
\usepackage{color}
\usepackage[colorlinks, linkcolor=blue, anchorcolor=green, citecolor=blue]{hyperref}

\allowdisplaybreaks

\newtheorem{theorem}{Theorem}[section]
\newtheorem{lemma}[theorem]{Lemma}
\newtheorem{corollary}[theorem]{Corollary}
\newtheorem{proposition}[theorem]{Proposition}
\newtheorem{remark}[theorem]{Remark}
\newtheorem{definition}[theorem]{Definition}
\numberwithin{equation}{section}
\begin{document}

\title{\Large\bf Robust Sparse Signal Recovery with Outliers: A Hard Thresholding Pursuit Approach Based on LAD
\footnotetext{\hspace{-0.35cm}
\endgraf $^\ast$\, Corresponding author.
\endgraf{$^1$ School of Mathematics and Statistics, Lanzhou University, Lanzhou, 730000, People's Republic of China \\
\indent $^2$ School of Mathematics and Statistics, Hubei Normal University, Huangshi, 435002, People's Republic of China}
\endgraf{ E-mail addresses: xujiao21@lzu.edu.cn, xujiao@hbnu.edu.cn (J. Xu), lp@lzu.edu.cn (P. Li), bzheng@lzu.edu.cn (B. Zheng)}}}
\author{Jiao Xu$^{1,2}$, Peng Li$^{1}$, Bing Zheng$^{1*}$}
\date{}

\maketitle
\textbf{Abstract.} Recovering a sparse signal from outlier-contaminated measurements is a fundamental challenge in many applications. While existing algorithms predominantly address scenarios with bounded noise or assume known signal sparsity, few methods tackle the more practical problem of sparse recovery from gross outliers without prior knowledge of sparsity. To bridge this gap, we study the sparsity-constrained Least Absolute Deviations (LAD) minimization problem. This paper proposes the Graded Fast Hard Thresholding Pursuit (GFHTP$_1$) algorithm with a quantile-truncated step size for $\ell_1$-loss minimization. In contrast to most state-of-the-art methods, our GFHTP$_1$ requires no prior knowledge of the signal's sparsity level. We establish a theoretical convergence analysis under mild conditions and further prove that an $s$-sparse signal can be recovered exactly within at most $s$ iterations. To our knowledge, these results provide the first efficient recovery guarantees for sparse signal reconstruction from outlier-corrupted measurements without a sparsity prior. Numerical experiments demonstrate that GFHTP$_1$ consistently outperforms competing algorithms in robustness to varying signal sparsity and outlier support size, while also achieving less computational time. 



\section{Introduction}\label{s1}
\hskip\parindent

\subsection{Problem Setup}\label{s1.1}
\hskip\parindent

Outliers are ubiquitous \cite{huber2011robust,zoubir2018robust,maronna2019robust}, for example sensor calibration \cite{li2016low}, face recognition \cite{de2003framework}, video surveillance \cite{li2004statistical}, and
their magnitudes can be arbitrarily large \cite{clason2012semismooth}.  In this paper, we investigate the problem of recovering sparse signals from linear measurements that are corrupted by a constant fraction of outliers with arbitrary magnitudes. Specifically, given a fixed measurement matrix $\mathbf{A}\in \mathbb{R}^{m\times n}$ (with $m\ll n$), our task is to identify an $s$-sparse $\mathbf{x}_0$ that satisfies:
\begin{equation}\label{problem1}
\mathbf{b}= \mathbf{Ax}_0+\boldsymbol{\eta},
\end{equation}
where $\mathbf{b}\in \mathbb{R}^m$ is the known measured response vector, and $\boldsymbol{\eta}$ is the unknown outliers with support $T$ and cardinality $|T|= pm\ll m$. The nonzero values of $\boldsymbol{\eta}$ are significantly larger than the nonzero components of the signal $\mathbf{x}_0$.

Our goal is to recover ${\bf x}_0$ from the observations \eqref{problem1}, noting that the outliers (or the residual) $\boldsymbol{\eta}=\mathbf{b}-\mathbf{Ax}_0$ is sparse, i.e., $\|\mathbf{b}-\mathbf{Ax}_0\|_0<m$. Here the notation $\|{\bf x}\|_0=|\{j\in[[n]]: |x_j|\neq 0\}|$ denotes the $\ell_0$-norm (count of nonzero entries). This leads to the following sparsity-constrained problem:
\begin{equation}\label{LADproblem.L0}
\min_{\mathbf{x}\in \mathbb{R}^n} \|\mathbf{b-Ax}\|_0, \ \mathrm{s.t.} \ \|\mathbf{x}\|_0\leq s
\end{equation}
with $s$ being an integer estimate of the sparsity level of $\mathbf{x}$. Since optimizing the $\ell_0$-norm is NP-hard, we relax the residual term to the $\ell_1$-norm to obtain the computationally feasible sparse-constrained LAD (Sparsity-LAD) model:
\begin{equation}\label{LADproblem}
\min_{\mathbf{x}\in \mathbb{R}^n} \|\mathbf{b-Ax}\|_1, \ \mathrm{s.t.} \ \|\mathbf{x}\|_0\leq s.
\end{equation}

It is well known that the traditional LS method $\min_{\mathbf{x}\in \mathbb{R}^n} \|\mathbf{b-Ax}\|_2$ is recognized for its optimality in scenarios where the measurement noise follows a Gaussian distribution. However, in practical applications, noise often exhibits non-Gaussian characteristics. The LS method's effectiveness is contingent upon assumptions regarding the noise's structure or standard deviation \cite{Bickel2009Lasso}, rendering it less reliable in the presence of impulsive noise, outliers, and other anomalies. In such cases, the LAD \cite{Bassett1978LAD} emerges as a valuable alternative. LAD is a statistical optimality criterion and a robust optimization technique that seeks to approximate data by minimizing the sum of the absolute values of the residuals \cite{Bloom1980LAD}. Dielman \cite{Dielman2005LAV} has shown that LAD outperforms the LS method in scenarios involving impulsive noise, sparse noise, and outliers, primarily due to its robust nature.  Unlike the LS method, which overweights large residuals, LAD treats all observations equally, making it suitable for data with arbitrary outliers.

Our primary objective in this paper is to develop an efficient, fast algorithm for solving this nonsmooth optimization problem by leveraging the hard thresholding pursuit (HTP) technique.

\subsection{Related Work}\label{s1.2}
\hskip\parindent

LAD has been applied in the development of robust methods across various domains, including statistical communities \cite{Dielman2005LAV}, compressive sensing signal reconstruction \cite{Paredes2011CS}, sparse representation-based face recognition \cite{Wagner2012FR}, channel estimation \cite{Jiang2013CE}, and image denoising \cite{Nikolova2004IN}. However, LAD solvers that are sparsely constrained still face critical limitations. 
Existing methods for sparse signal recovery under outliers via LAD can be categorized based on their ability to address three core challenges: outliers robustness, unknown sparsity, and computational efficiency. Below, we review relevant work and highlight their limitations relative to our research goals.

Based on the underlying regularization strategy, existing methods fall into two categories: (i) regularized-minimization methods (the relaxed methods), which correspond to relaxed convex or nonconvex regularization minimization; (ii) hard threshold-based methods, which correspond to the $\ell_0$ regularization (or sparse constraint).

 We first review the regularized-minimization methods. Several scholars have also proposed convex or nonconvex relaxation methods for solving the LAD problem. For instance, Yang and Zhang \cite{YangZ2011LAD} first introduced the regularized (or penalized) LAD (RLAD/PLAD) model, which incorporates both the $\ell_1$ regularized function $\|\mathbf{x}\|_1$ and the $\ell_1$ loss function $\|\mathbf{b-Ax}\|_1$ into the objective function. They designed a solving algorithm via the alternating direction method of multipliers (ADMM).
In $2013$, Wang \cite{Wang2013LAD} conducted a theoretical analysis of this convex model, which is capable of handling cases where the error distribution is unknown or exhibits heavy tails, even for Cauchy distributions. The corresponding estimator demonstrates near-optimal performance with high probability. Notably, Li et al. \cite{Li2020LAD} have proposed a nonconvex minimization method with LAD constraint (Nonconvex LAD), which aims to solve an $\ell_1-\alpha\ell_2$ minimization problem. They also
contributed to the theoretical analysis of the $\ell_1-\alpha\ell_2$ minimization model. 
However, the theoretical analysis provided for the proposed model is not applicable to cases with outliers. Recently, Xu et al. \cite{Xu2024RLAD} have presented a theoretical analysis for the unconstrained version of the nonconvex $\ell_1-\alpha\ell_2$ minimization model (Nonconvex RLAD), which is valid when observations are corrupted by the outliers. Nevertheless, numerical experiments reveal that the model's performance degrades significantly when the outlier rate is high (see \cite[Table 5]{Xu2024RLAD}). 

To better address the drawback, another line focuses on a hard thresholding-based algorithm, which primarily focuses on sparse recovery via dense noise environments. To solve Sparsity LAD problem \eqref{LADproblem}, 
Li et al. \cite{Li2023AIHT} firstly introduced an adaptive iterative hard thresholding (AIHT) algorithm, which performs subgradient descent followed by a hard thresholding operator $\mathcal{H}_s(\mathbf{x})$. This operator retains the largest $s$ elements of $\mathbf{x}$ in magnitude and sets the remainder to zero.
Their AIHT algorithm is globally convergent under bounded noise but fails in the presence of outliers. Additionally, it requires prior knowledge of the sparsity $s$—a major drawback in practical applications. 
Recently, Xu et al. \cite{Xu2024PSGD} proposed the Projected Subgradient Descent (PSGD) algorithm to solve the Sparsity LAD problem \eqref{LADproblem} in the presence of both outliers and bounded noise, providing a full convergence analysis. Nevertheless, the selection of the step size in this algorithm may not be optimal, as it relies on the actual signal--an issue that can lead to inaccuracies in its application. These sparsity-dependent methods, however, fail to fully exploit the constraint $\|\mathbf{x}\|_0 \leq s$. Consequently, their performance degrades when the sparsity level $s$ is high, often preventing exact recovery of the true signal's support.

\subsection{Motivation and Contributions}\label{s1.3}
\hskip\parindent

While existing hard thresholding-based algorithms can manage outliers, they possess significant limitations. First, these methods require knowledge of the true sparsity level—a prior that is typically unavailable in practice. Second, their robustness deteriorates when the sparsity level is high. Third, focusing on the specific PSGD algorithm \cite{Xu2024PSGD}, its step size strategy depends on the actual signal, making it difficult to adapt to real-world applications. Furthermore, the absence of a concrete stopping criterion in \cite{Xu2024PSGD} presents a practical hurdle. Consequently, there is a clear need for a more adaptable step size rule and a practical, well-defined stopping criterion.

To address these challenges, we propose a novel algorithm solving Sparsity LAD \eqref{LADproblem}. Our approach is built on three key strategies: (i) Enhanced Robustness to Sparsity: We employ a two-phase procedure: a solving step that identifies a candidate support set, followed by a pursuit step that refines the signal estimate within this identified support.(ii) Signal-Independent Step Size: To eliminate the reliance on prior knowledge of the true signal, we introduce a novel truncated step size rule based on the quantile. (iii) Elimination of the Sparsity Prior: We further incorporate a graded hard-thresholding mechanism within the pursuit step, removing the need for the true sparsity level as an input parameter.

Motivated by the aforementioned discussion, we propose both  {\bf fast hard thresholding pursuit (FHTP$_1$)} and   {\bf graded fast hard thresholding pursuit (GFHTP$_1$)} algorithms for the sparsity-constrained LAD \eqref{LADproblem} with the observations \eqref{problem1}. This paper's contributions are as follows:
\begin{itemize}
  \item[(a)] {\bf Parameter-Free Algorithm with Unknown Sparsity}: GFHTP$_1$ integrates FHTP's inner-iteration acceleration with GHTP's graded support growth (support size $= k$ at outer iteration $k$), eliminating the need for prior knowledge of 
$s$. It uses a {\bf truncated adaptive step size} (dependent only on small residual components, not the true signal) to suppress outliers, filling the gap of HTP-based methods for LAD.
  \item[(b)] {\bf Rigorous Convergence Analysis}: We analyze convergence for two signal types:
\begin{itemize}
  \item[$\bullet$] For general $s$-sparse signals, we establish a linear error bound under the restricted isometry property (RIP$_1$).
  \item[$\bullet$] For `flat' signals satisfying $x_1^*\leq \lambda x_s^*$ ($\lambda\geq 1, x_j^*$ is the non-increasing rearrangement of $|\mathbf{x}_0|$), we prove {\bf exact recovery at the 
$s$-th outer iteration} ($\mathbf{x}^s= \mathbf{x}_0$) with high probability.
\end{itemize}
  \item[(c)] {\bf Practical Step Size, Stopping Criterion, and Superior Performance}: We handle outliers by incorporating a truncated adaptive step size (based on\\ $\|(\mathbf{b-Ax}^{k})\odot (\mathbb{I} _{\{|b_i-(\mathbf{Ax}^{k})_i|\leq \theta_\tau(|\mathbf{b-Ax}^{k}|)\}})_{i=1}^m\|_1\leq \epsilon_{\mathrm{outer}}$). Following this idea, we also design a highly efficient stopping criterion that promotes fast and precise convergence. Numerical experiments demonstrate GFHTP$_1$ outperforms PSGD and AIHT in terms of success rates.
  \item[(d)]{\bf Key Skill and Proposition for Our Theoretical Analysis.} In order to remove the outliers in the theoretical analysis, we establish a key sandwich inequality, which provides the lower and upper bounds of $\|(\mathbf{b-Ax}^{k})\odot (\mathbb{I} _{\{|b_i-(\mathbf{Ax}^{k})_i|\leq \theta_\tau(|\mathbf{b-Ax}^{k}|)\}})_{i=1}^m\|_1$.
 Additionally, we also introduce a key proposition, which shows that the support $S^k$ in the $k$-iteration is a subset of the true support $S$ of the signal $\mathbf{x}_0$. This key proposition enables our theoretical results for `flat' signals. 
\end{itemize}

Detailed comparisons with existing methods are presented in Table \ref{ComparisonDifferentMethods}. Our algorithms and main theoretical results---Theorem \ref{thm3.1}, Corollary \ref{FHTP_theoretical}, and Theorem \ref{propGFHTP} in Section \ref{s7}---offer three key advantages: (i) Our convergence guarantees hold even under outlier contamination, whereas Nonconvex LAD and AIHT are effective only for bounded noise, and RLAD/Nonconvex RLAD only handle symmetric outliers. (ii) Our GFHTP$_1$ eliminates the sparsity-dependent requirements present in both AIHT and PSGD.

\begin{table}[!ht]
\centering
\caption{Comparison between our results and existing methods in LAD. Here, ``$-$'' indicates no available result, and $\mathbf{e}$ denotes bounded noise. Symmetric noise satisfies $\mathbb{P}(e_i>0)=\mathbb{P}(e_i<0)=0.5$ (see \cite[Remark 2.1]{Xu2024RLAD}).} 
\label{ComparisonDifferentMethods}
\resizebox{\textwidth}{!}{%
\begin{tabular}{|c|c|c|c|c|}
\hline
\makecell[c]{Objective} & 
\makecell[c]{Methods} & 
\makecell[c]{Noise} & 
\makecell[c]{Sparsity} &
\makecell[c]{No. of Iters.} \\
\hline
\multirow{2}{*}{Regularization} 
& RLAD \cite{Wang2013LAD}
  & Symmetric noise 
  & No  
  & $-$  \\
  & Noncovex LAD \cite{Li2020LAD}
  & Bounded noise 
  & No  
  & $-$ \\
& Nonconvex RLAD \cite{Xu2024RLAD}
 & Symmetric noise 
  & No  
  & $-$  \\
\hline
\multirow{4}{*}{Sparsity Constraint} 
& AIHT \cite{Li2023AIHT} 
  & Bounded Noise 
  & Need      
  & $\mathcal{O}(\log(1/\varepsilon))$ \\[0.3em]
& PSGD \cite{Xu2024PSGD} 
  & \makecell[c]{Bounded Noise\\Outliers} 
  & Need       
  & \makecell[c]{$\mathcal{O}(\log(1/\varepsilon))$\\$\mathcal{O}(\log(1/\varepsilon))$} \\
  & FHTP$_1$ (Ours) 
  & Outliers 
  & Need       
  & $\mathcal{O}(\log(1/\varepsilon))$ \\
& GFHTP$_1$ (Ours) 
  & Outliers 
  & No       
  & $\mathcal{O}(s)$ \\
\hline
\end{tabular}%
}
\end{table}

\subsection{Organization and Notations}\label{s1.4}
\hskip\parindent

The remainder of this paper is structured as follows. Section \ref{s3} displays the solving algorithms. Section \ref{s7} delves into the theoretical analysis of the GFHTP$_1$ algorithm. 
Section \ref{sec3roadmap} gives the roadmap and keystone of our proofs. 
Section \ref{s4} examines the numerical performance of the proposed algorithms through experiments. Section \ref{s5} concludes the paper with a summary of the findings. Appendix \ref{secappendixa} presents the proof of the general sparse signal recovery case. Appendix \ref{secappendixb} provides the theoretical proof for the general sparse signal recovery case.

Throughout this paper, we use the following notations. Matrices are denoted by boldface capital letters, such as $\mathbf{A}$, while vectors are represented by boldface lowercase letters, for example, $\mathbf{a}$. Scalars are indicated by regular lowercase letters, such as $a$. The sign function $\mathrm{sign}(\cdot)$ is defined as $\mathrm{sign}(a)= a/|a|$ for $a\neq 0$, with $\mathrm{sign}(0)= 0$. For any positive integer $n$, let $[[n]]$ represent the set $\{1,\cdots, n\}$. The notation $\mathbf{x}_\Omega\in \mathbb{R}^n$ refers to a vector where its elements are equal to $\mathbf{x}$ for indices within the set $\Omega$, and zero otherwise. Let $\Omega^c$ denote the complement of the index set $\Omega$, which is defined as $\Omega^c= [[n]]\setminus\Omega$. Let $\mathbf{A}_W$ denote the submatrix of $\mathbf{A}$ obtained by keeping the rows of $\mathbf{A}$ with indices in the set $W$. The indicator function $\mathbb{I}_{B}= 1$ if the event $A$ holds, and $\mathbb{I}_{B}= 0$ otherwise. The notation $\odot$ refers to the Hadamard product, and $\Phi$ is the cumulative distribution function of the standard Gaussian distribution. We use $\theta_{\tau}$ to denote the $\tau$-th quantile. We use $\mathbf{x}^*\in \mathbb{R}_+^n$ to represent the non-increasing rearrangement of the original signal $\mathbf{x}_0\in \mathbb{R}^n$, meaning $x_1^*\geq x_2^*\geq \cdots\geq x_n^*\geq 0$. There exists a permutation $\nu$ of $[[n]]$ such that $x_j^*= |x_{\nu(j)}|$ for all $j\in [[n]]$. For two index sets $\Omega$ and $\Omega'$, $\Omega\bigtriangleup \Omega'$ denotes the symmetric difference of these sets, which is the union of their respective differences, i.e., $\Omega\bigtriangleup \Omega'= (\Omega\setminus \Omega')\cup(\Omega'\setminus \Omega)$.

\section{Solving Algorithms}\label{s3}
\hskip\parindent

In this subsection, we design two efficient solving algorithms for the model \eqref{LADproblem}. The first one needs the sparsity prior, while the second one addresses the limitation of requiring prior sparsity information. 

\subsection{Parameter Description}
\hskip\parindent

Before introducing the algorithms, we clarify the key parameters to ensure reproducibility and interpretability:
\begin{itemize}
  \item[(i)] $\tau$ (quantile parameter): The quantile used to truncate outliers, satisfying $p<\tau<1-p$ ($p$ is the outliers rate). Its role is to filter out large residual components (outliers) when calculating the step size, avoiding interference with iterative updates.
  \item[(ii)] $\theta_{\tau}(|\mathbf{b-Ax}^k|)$: The $\tau$-quantile of the absolute residual vector $|\mathbf{b-Ax}^k|$, calculated based on the empirical distribution of the residual.
  \item[(iii)] $\mu_{k,l}$ (adaptive step size coefficient): A positive constant determining the step size scale, recommended to be initialized to 6 (verified by numerical experiments in Section \ref{s4.2}) and adjusted within the range derived from theoretical conditions (see Theorem \ref{thm3.1} and Remark \ref{remark2.3}).
  \item[(iv)] $\mathrm{MaxIt}$: Maximum number of outer iterations, recommended to be set to $\mathrm{ceil}(m/2)$ ($m$ is the measurement dimension) to ensure the support set covers the true sparsity.
  \item[(v)] $\epsilon_{\mathrm{outer}}, \epsilon_{\mathrm{inner}}$: Termination thresholds for outer/inner iterations, controlling the convergence accuracy.
\end{itemize}

\subsection{Fast Hard Thresholding Pursuit (FHTP$_1$) Algorithm}
\hskip\parindent

In this subsection, we design a fast algorithm for solving the nonsmooth sparsity-constrained LAD \eqref{LADproblem}.  To solve this optimization problem, we adopt the following alternating minimization scheme:
\begin{equation}\label{AlternatingMinimization}
\begin{cases}
S^{k+1}=\arg\min_{|S|\leq s}\|\mathbf{A}\mathbf{x}_{S}-\mathbf{b}\|_1,~\text{(Find~the~Candidate~Support)}&\\
\mathbf{x}^{k+1}=\arg\min_{\mathbf{x}:\mathrm{supp}(\mathbf{x})\subset S^{k+1}}\|\mathbf{Ax}-\mathbf{b}\|_1,\text{(Update~the~Sparse~Signal)}
\end{cases}
\end{equation}
The subproblems in the alternating minimization scheme \eqref{AlternatingMinimization} are solved either exactly or approximately, as detailed below:
\begin{itemize}
\item[(i)] Finding the Candidate Support: We update the candidate support via the subgradient descent followed by a hard thresholding operator \begin{equation}\label{FHTP1.CandidateSupport}
\mathbf{u}^{k+1,1}:= \mathcal{H}_{s}(\mathbf{x}^k+t_{k+1,0}\mathbf{A}^\top \mathrm{sign}(\mathbf{b}-\mathbf{Ax}^k)), \quad S^{k+1}= \mathrm{supp}(\mathbf{u}^{k+1,1}),
\end{equation}
where $t_{k+1,0}>0$ is a step size, and $-\mathbf{A}^\top \mathrm{sign}(\mathbf{b}-\mathbf{Ax}^k)$ is the subgradient of the objective function $\|\mathbf{b}-\mathbf{A}\mathbf{x}^k\|_1$ at the current point $\mathbf{x}^{k}$. 
\item[(ii)] Updating the Sparse Signal: The subproblem has no closed-form solution, so we solve it via subgradient descent with restriction to the given support
\begin{equation}\label{FHTP1.PursuitStep}
\mathbf{u}^{k+1,l+1}:= (\mathbf{u}^{k+1,l}+t_{k+1,l}\mathbf{A}^\top \mathrm{sign}(\mathbf{b}-\mathbf{Au}^{k+1,l}))_{S^{k+1}}, \quad \mathbf{x}^{k+1}= \mathbf{u}^{k+1,L+1}.
\end{equation}
\end{itemize}

In summary, our solving algorithm contains the following two iterative steps:
\begin{equation}\label{FHTP1}
\begin{cases}
\mathbf{u}^{k+1,1}:= \mathcal{H}_{s}(\mathbf{x}^k+t_{k+1,0}\mathbf{A}^\top \mathrm{sign}(\mathbf{b}-\mathbf{Ax}^k)), \quad S^{k+1}= \mathrm{supp}(\mathbf{u}^{k+1,1}), &\\
\mathbf{u}^{k+1,l+1}:= (\mathbf{u}^{k+1,l}+t_{k+1,l}\mathbf{A}^\top \mathrm{sign}(\mathbf{b}-\mathbf{Au}^{k+1,l}))_{S^{k+1}}, \quad \mathbf{x}^{k+1}= \mathbf{u}^{k+1,L+1}.
\end{cases}
\end{equation}
Note that the idea of the second step is similar to that of FHTP \cite{Foucart2011HTP}, which solves an LS with a constraint on the given support set. We adjust FHTP to suit the nonsmooth model \eqref{LADproblem}. Therefore we call this algorithm the Fast Hard Thresholding Pursuit (FHTP$_1$) for the $\ell_1$ loss function.
The algorithm is presented in Algorithm \ref{algorithm1.1}.

\begin{algorithm}
\caption{Fast hard thresholding pursuit (FHTP$_1$) for solving \eqref{problem1}}
\label{algorithm1.1}
\begin{algorithmic}[1]
\Require $\mathbf{A}, \mathbf{b}, \mathbf{x}^0, S^0= \mathrm{supp}(\mathbf{x}^0), s, \mathrm{MaxIt}, L, \epsilon_{\mathrm{outer}}, \epsilon_{\mathrm{inner}}$
\Ensure $\mathbf{x}$
        \State {\bf Outer loop:}  \While {$0\leq k\leq \mathrm{MaxIt}$ and $S^{k}\neq S^{k-1}$ and $\|(\mathbf{b-Ax}^{k})\odot (\mathbb{I} _{\{|b_i-(\mathbf{Ax}^{k})_i|\leq \theta_\tau(|\mathbf{b-Ax}^{k}|)\}})_{i=1}^m\|_1> \epsilon_{\mathrm{outer}}$}
        \State Compute $t_{k+1,0}= \mu_{k+1,0}\sqrt{\frac{\pi}{2}}\|(\mathbf{b-Ax}^{k})\odot (\mathbb{I}_{\{|b_i-(\mathbf{Ax}^{k})_i|\leq \theta_\tau(|\mathbf{b-Ax}^{k}|)\}})_{i=1}^m\|_1$.
        \State Compute $\mathbf{u}^{k+1,1}:= \mathcal{H}_s(\mathbf{x}^{k}+t_{k+1,0}\mathbf{A}^\top \mathrm{sign}(\mathbf{b}-\mathbf{Ax}^{k}))$ and $S^{k+1}= \mathrm{supp}(\mathbf{u}^{k+1,1})$. Set $\mathbf{u}^{k+1,0}=\mathbf{x}^{k}$.
        \State {\bf Inner loop:} \While {$1\leq l\leq L$ and $\|\mathbf{u}^{k+1,l}-\mathbf{u}^{k+1,l-1}\|_2/\|\mathbf{u}^{k+1,l-1}\|_2> \epsilon_{\mathrm{inner}}$}
\State Compute $t_{k+1,l}= \mu_{k+1,l}\sqrt{\frac{\pi}{2}}\|(\mathbf{b-Au}^{k+1,l})\odot (\mathbb{I}_{\{|b_i-(\mathbf{Au}^{k+1,l})_i|\leq \theta_\tau(|\mathbf{b-Au}^{k+1,l}|)\}})_{i=1}^m\|_1$.		
\State  Compute $\mathbf{u}^{k+1,l+1}:= (\mathbf{u}^{k+1,l}+t_{k+1,l}\mathbf{A}^\top \mathrm{sign}(\mathbf{b}-\mathbf{Au}^{k+1,l}))_{S^{k+1}}$.
        \State Set $l:= l+1$.
		\EndWhile \ and output $\mathbf{x}^{k+1}:= \mathbf{u}^{k+1,L+1}$.
        \State Set $k:= k+1$.
        \EndWhile
\end{algorithmic}
\end{algorithm}

Next, we design a stopping criterion. We set $\mathrm{MaxIt}= \mathrm{ceil}(m/2)$ and incorporate an appropriate stopping criterion, which is defined as follows:
\begin{itemize}
  \item[(i)] For the inner iteration: $\|\mathbf{u}^{k+1,l+1}-\mathbf{u}^{k+1,l}\|_2/\|\mathbf{u}^{k+1,l}\|_2\leq \epsilon_{\mathrm{inner}}= 10^{-8}$;
  \item[(ii)] For the outer iteration: $S^{k+1}= S^k$ or $\|(\mathbf{b-Ax}^{k+1})\odot (\mathbb{I}_{\{|b_i-(\mathbf{Ax}^{k+1})_i|\leq \theta_\tau(|\mathbf{b-Ax}^{k+1}|)\}})_{i=1}^m\|_1\leq \epsilon_{\mathrm{outer}}= 10^{-4}$.
\end{itemize} 

To conclude this subsection, we analyze the time complexity of Algorithm \ref{algorithm1.1}. It is mainly dominated by computing the hard thresholding operator $\mathcal{H}_s(\cdot)$ and $t_{k+1,l}$, and updating $\mathbf{u}^{k+1,l+1}$. We obtain that the time complexity for computing $\mathcal{H}_s(\cdot)$ is $\mathcal{O}(n\log s)$. Based on matrix multiplication and $\tau$-quantile, the time complexity of calculating $t_{k+1,l}$ and $\mathbf{u}^{k+1,l+1}$ is $\mathcal{O}(sm + m \log m)$ and $\mathcal{O}(mn)$, respectively. Therefore the time complexity is $\mathcal{O}(k_{\mathrm{outer}}(L(mn+m\log m)+n\log s))$, where $k_{\mathrm{outer}}$ denotes the number of outer iterations in Algorithm \ref{algorithm1.1}.

\subsection{Graded Fast Hard Thresholding Pursuit (GFHTP$_{1}$) Algorithm}
\hskip\parindent

The FHTP$_1$ algorithm relies on prior knowledge of sparsity $s$, which limits its practical application. To address this, we propose a graded algorithm, which constructs a sequence of $(k+1)$-sparse vectors $(\mathbf{x}^{k+1})$ with an index set that grows with each iteration:
\begin{equation}\label{GFHTP}
\begin{cases}
 \mathbf{u}^{k+1,1}:= \mathcal{H}_{k+1}(\mathbf{x}^k+t_{k+1,0}\mathbf{A}^\top \mathrm{sign}(\mathbf{b}-\mathbf{Ax}^k)), \quad S^{k+1}= \mathrm{supp}(\mathbf{u}^{k+1,1}),& \\
 \mathbf{u}^{k+1,l+1}:= (\mathbf{u}^{k+1,l}+t_{k+1,l}\mathbf{A}^\top \mathrm{sign}(\mathbf{b}-\mathbf{Au}^{k+1,l}))_{S^{k+1}}, \quad \mathbf{x}^{k+1}= \mathbf{u}^{k+1,L+1}.
\end{cases}
\end{equation}
The algorithm is described in Algorithm \ref{algorithm1.2}.  Note that our graded strategy is similar to that of GHTP \cite{Bouchot2016GHTP}, therefore we call our algorithm the Graded FHTP (GFHTP$_1$) for the $\ell_1$ loss function.

\begin{algorithm}
\caption{Graded fast hard thresholding pursuit (GFHTP$_1$) for solving \eqref{problem1}}
\label{algorithm1.2}
\begin{algorithmic}[1]
\Require $\mathbf{A}, \mathbf{b}, \mathbf{x}^0, S^0= \mathrm{supp}(\mathbf{x}^0), \mathrm{MaxIt}, L, \epsilon_{\mathrm{outer}}, \epsilon_{\mathrm{inner}}$
\Ensure $\mathbf{x}$
        \State {\bf Outer loop:}  \While {$k\leq \mathrm{MaxIt}$ and $\|(\mathbf{b-Ax}^{k})\odot (\mathbb{I}_{\{|b_i-(\mathbf{Ax}^{k})_i|\leq \theta_\tau(|\mathbf{b-Ax}^{k}|)\}})_{i=1}^m\|_1> \epsilon_{\mathrm{outer}}$}
        \State Compute $t_{k+1,0}= \mu_{k+1,0}\sqrt{\frac{\pi}{2}}\|(\mathbf{b-Ax}^{k})\odot (\mathbb{I}_{\{|b_i-(\mathbf{Ax}^{k})_i|\leq \theta_\tau(|\mathbf{b-Ax}^{k}|)\}})_{i=1}^m\|_1$.
        \State Compute $\mathbf{u}^{k+1,1}:= \mathcal{H}_{k+1}(\mathbf{x}^k+t_{k+1,0}\mathbf{A}^\top \mathrm{sign}(\mathbf{b}-\mathbf{Ax}^k))$ and $S^{k+1}= \mathrm{supp}(\mathbf{u}^{k+1,1})$. Set $\mathbf{u}^{k+1,0}=\mathbf{x}^{k}$.
        \State {\bf Inner loop:} \While {$1\leq l\leq L$ and $\|\mathbf{u}^{k+1,l}-\mathbf{u}^{k+1,l-1}\|_2/\|\mathbf{u}^{k+1,l-1}\|_2> \epsilon_{\mathrm{inner}}$}
\State Compute $t_{k+1,l}= \mu_{k+1,l}\sqrt{\frac{\pi}{2}}\|(\mathbf{b-Au}^{k+1,l})\odot (\mathbb{I}_{\{|b_i-(\mathbf{Au}^{k+1,l})_i|\leq \theta_\tau(|\mathbf{b-Au}^{k+1,l}|)\}})_{i=1}^m\|_1$.		
\State  Compute $\mathbf{u}^{k+1,l+1}:= (\mathbf{u}^{k+1,l}+t_{k+1,l}\mathbf{A}^\top \mathrm{sign}(\mathbf{b}-\mathbf{Au}^{k+1,l}))_{S^{k+1}}$.
        \State Set $l:= l+1$.
		\EndWhile \ and output $\mathbf{x}^{k+1}:= \mathbf{u}^{k+1,L+1}$.
        \State Set $k:= k+1$.
        \EndWhile
\end{algorithmic}
\end{algorithm}

For the GFHTP$_1$ algorithm, we adopt its stopping criterion by removing the condition $S^{k+1}= S^k$ from the stopping criterion of the FHTP$_1$ algorithm.

\section{Our Theoretical Results}\label{s7}
\hskip\parindent

In this section, we present the convergence properties of the proposed algorithms, providing clear theoretical support for their effectiveness in sparse signal recovery with outliers.
We first derive error bounds for general sparse signals, then prove exact recovery result for specific sparse signals. 

\subsection{Foundational Assumption: Restricted 1-Isometry Property (RIP$_1$)}
\hskip\parindent

\begin{definition}\label{RIP_def}(\cite[Inequation (1.8)]{Chartrand2008RIP})
A matrix $\mathbf{A}\in \mathbb{R}^{m\times n}$ is said to satisfy the RIP$_1$ of order $s$ if there exists a small constant $\delta_s\in (0,1)$ such that the inequality
\begin{equation}\label{RIP.eq1}
(1-\delta_s)\|\mathbf{x}\|_2\leq \sqrt{\frac{\pi}{2}}\|\mathbf{Ax}\|_1\leq (1+\delta_s)\|\mathbf{x}\|_2
\end{equation}
holds for all vectors $\mathbf{x}\in \mathbb{R}^n$ with sparsity not exceeding $s$. The smallest constant $\delta_s$ that satisfies this inequality is known as the restricted $1$-isometry constant (RIC$_1$). 
\end{definition}
What we should point out is that Gaussian random matrices satisfy RIP$_1$ with high probability (Lemma \ref{RIP}), which is foundational for subsequent theorems.

\subsection{Error Estimation for General Sparse Signals}
 \hskip\parindent 

Firstly, we establish the convergence analysis of the proposed algorithms for a general sparse signal. 
\begin{theorem}\label{thm3.1}
Given an $s$-sparse vector $\mathbf{x}_0\in \mathbb{R}^n$ and the measurements $\mathbf{b}= \mathbf{Ax}_0+\boldsymbol{\eta}$, where $\boldsymbol{\eta}$ is a vector of outliers with support $T$ and cardinality $|T|= pm$. Assume that the fraction of outliers $p$ satisfies $p<\frac{1}{2}-\frac{\delta_{2k+s-1}}{1-\delta_{2k+s-1}}$. Let $\{\mathbf{x}^{k}\}_{k=1}^{\infty}$ be the sequence produced by the GFHTP$_1$ algorithm, with adaptive step size 
$$t_{k,l}= \mu_{k,l}\sqrt{\frac{\pi}{2}}\|(\mathbf{b-Au}^{k,l})\odot (\mathbb{I}_{\{|b_i-(\mathbf{Au}^{k,l})_i|\leq \theta_\tau(|\mathbf{b-Au}^{k,l}|)\}})_{i=1}^m\|_1,$$
where $\mu_{k,l}$ satisfies the following inequality
\begin{align}\label{GFHTP_condition}
0<\rho_{k,l}:=&1+\tau^2\left(\Phi^{-1}\left(\frac{1+\tau+p}{2}\right)+\epsilon\right)^2(1+\delta_{2k+s-1})^2\mu_{k,l}^2 \nonumber\\
&- 2c_k\sqrt{\frac{2}{\pi}}\left(\tau-\frac{|T_1^{k,l}|}{m}\right)(1-\delta_{2k+s-1})\mu_{k,l}<\frac{1}{3}.
\end{align}
Here $c_k= (2-2p)(1-\delta_{2k+s-1})-(1+\delta_{2k+s-1})$, $T_1^{k,l}= T\cap \Gamma^{k,l}, \Gamma^{k,l}=\{i:|b_i-(\mathbf{Au}^{k,l})_i|\leq \theta_\tau(|\mathbf{b-Au}^{k,l}|)\}$, and $\epsilon$ is a small constant. Then the $k$-th iterate $\mathbf{x}^k$ satisfies:
\begin{equation}\label{errorestimation}
\|\mathbf{x}_0-\mathbf{x}^{k}\|_2^2\leq \left(\frac{(\rho_{k})^{L+1}(1-3\rho_{k})+2\rho_{k}}{1-\rho_{k}}\right)\|\mathbf{x}_0-\mathbf{x}^{k-1}\|_2^2,
\end{equation}
where $\mathbf{u}^{k,L+1}:= \mathbf{x}^k, \mathbf{u}^{k,0}:= \mathbf{x}^{k-1}$, $\rho_k=\max_{l=0}^{L}\rho_{k,l}$, and $k\geq s$. 
\end{theorem}

Then, we show a corollary which displays the convergence of FHTP$_1$ algorithm.
\begin{corollary}\label{FHTP_theoretical}
Let $\mathbf{x}_0\in \mathbb{R}^n$ be $s$-sparse and $\mathbf{b}= \mathbf{Ax}_0+\boldsymbol{\eta}$ for some $\boldsymbol{\eta}$ with $T=\mathrm{supp}(\boldsymbol{\eta})$ and $|T|= pm$. Assume that the fraction of outliers $p$ satisfies $p<\frac{1}{2}-\frac{\delta_{3s}}{1-\delta_{3s}}$. Let $(\mathbf{x}^{k})$ be the sequence generated by FHTP$_1$ and adaptive step size satisfying the same form as Theorem \ref{thm3.1} (replacing $\delta_{2k+s-1}$ with $\delta_{3s}$). Then
\begin{equation}\label{errorestimation1}
\|\mathbf{x}_0-\mathbf{x}^{k}\|_2^2\leq \left(\frac{(\rho_{k})^{L+1}(1-3\rho_{k})+2\rho_{k}}{1-\rho_{k}}\right)\|\mathbf{x}_0-\mathbf{x}^{k-1}\|_2^2,
\end{equation}
where $\rho_k=\max_{l=0}^{L}\rho_{k,l}$, and $c_k= (2-2p)(1-\delta_{3s})-(1+\delta_{3s})$. Moreover, we obtain
\begin{equation}\label{errorestimation}
\|\mathbf{x}_0-\mathbf{x}^{k}\|_2\leq \rho^{k}\|\mathbf{x}_0-\mathbf{x}^{0}\|_2
\end{equation}
with $\rho:=\max_k \sqrt{\frac{(\rho_{k})^{L+1}(1-3\rho_{k})+2\rho_{k}}{1-\rho_{k}}}$. With the initialization $\mathbf{x}^0=\mathbf{0}$, we can reconstruct an $s$-sparse signal $\mathbf{x}_0$ with an accuracy of  $\|\mathbf{x}^{k^*} - \mathbf{x}_0\|_2\leq \varepsilon$ after 
$$
k^{*}=\left\lceil \log_{\frac{1}{\rho}}\frac{\|\mathbf{x}_0\|_2}{\varepsilon}\right\rceil$$
iterations.

\noindent {\bf Note}: FHTP$_1$ uses fixed sparsity $s$, so the RIP$_1$ order is $3s$. 
\end{corollary}

Next, we give a remark which checks the sufficient condition \eqref{GFHTP_condition} in Theorem \ref{thm3.1}.
\begin{remark}\label{remark2.3}
It seems that the sufficient condition \eqref{GFHTP_condition} is complex and strict. In fact, this condition can be met. Notice that the sufficient condition \eqref{GFHTP_condition} provides a quadratic inequality $ax^2-bx+c<0$ for $x=\mu_{k,l}$ with $a=\tau^2\left(\Phi^{-1}\left(\frac{1+\tau+p}{2}\right)+\epsilon\right)^2(1+\delta_{2k+s-1})^2$, $b=2c_k\sqrt{\frac{2}{\pi}}\left(\tau-\frac{|T_1^{k,l}|}{m}\right)(1-\delta_{2k+s-1})$ and $c=2/3$. For this quadratic inequality to have real solutions, it is necessary to ensure that $b^2-4ac>0$; when this condition holds, the inequality has solutions given by $(b-\sqrt{b^2-4ac})/2a<\mu_{k,l}< (b+\sqrt{b^2-4ac})/2a$. We set $\epsilon= 0.001, \delta_{2k+s-1}= 0.01, |T_1^{k,l}|/m= 0.001$, $\tau= 0.1:0.001:0.7$, $p= 0.001:0.0001:0.5$. We can take $|T_1^ {k,l}|/m\leq0.001$. The selectable range for the outlier proportion $p$ is shown in Figure \ref{max_p} (a). Especially, 
\begin{itemize}
\item [\rm{(i)}] when $\tau= 0.5, p= 0.05$, we can take $1.3695<\mu_{k,l}< 3.3362$ to make this condition \eqref{GFHTP_condition} true;
\item [\rm{(ii)}] we can set $\tau= 0.1, p= 0.2$, then we can find that when $8.7136<\mu_{k,l}< 50.2541$, the condition \eqref{GFHTP_condition} is satisfied.
\end{itemize}

\begin{figure}
\centering
\subfigure[General Sparse Signal]
{
\begin{minipage}{7cm}
\centering
\includegraphics[scale = 0.48]{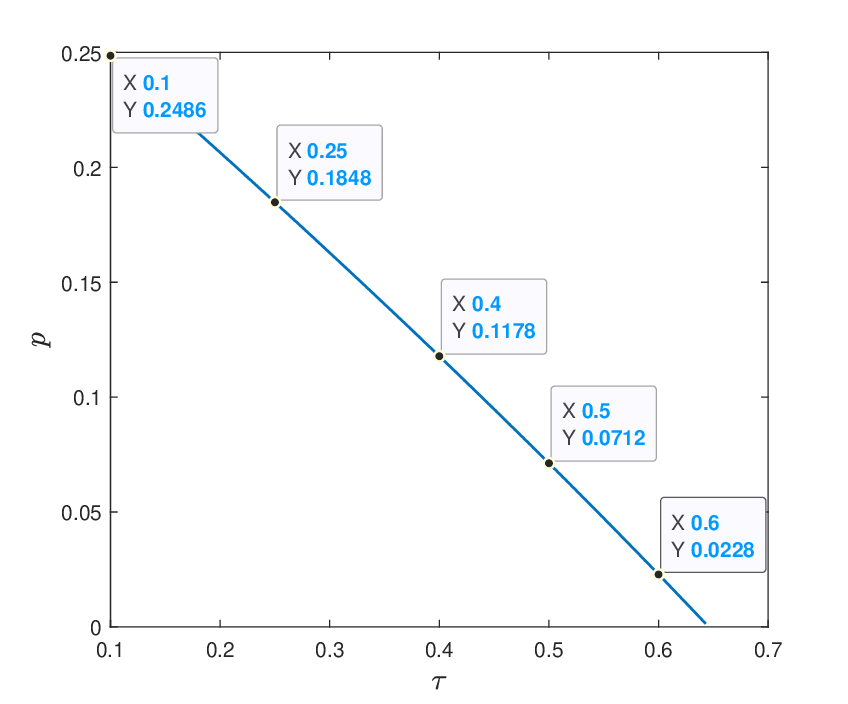}%
\end{minipage}
}
\subfigure[`Flat' Sparse Signal]
{
\begin{minipage}{7cm}
\centering
\includegraphics[scale = 0.48]{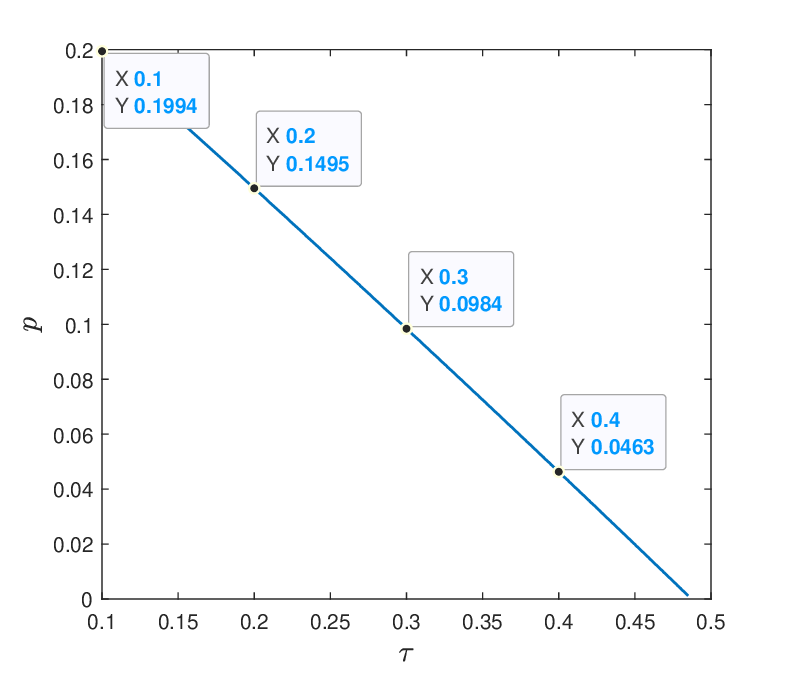}%
\end{minipage}
}
\caption{The maximum value of the outliers proportion 
$p$.}
\label{max_p}
\end{figure}
\end{remark}

\subsection{Exact Recovery for Structured Sparse Signals}
\hskip\parindent 
Next, we establish the convergence analysis of the proposed GFHTP$_1$ algorithm for special sparse signals, namely 'flat' signals.
\begin{theorem}\label{propGFHTP}
Consider an $s$-sparse vector $\mathbf{x}_0 \in \mathbb{R}^n$ satisfying $x_1^*\leq \lambda x_s^*$ ($\lambda\geq 1$, e.g., `flat' signals with $\lambda= 1$). Assume that the sensing matrix $\mathbf{A}$ is an $m\times n$ matrix whose entries are independently and identically distributed (i.i.d.) Gaussian random variables, i.e., ${a}_{ij}\sim \mathcal{N}(0, \frac{1}{m^2})$ with $m\geq c'_1s\log n$. Let the measurements $\mathbf{b}= \mathbf{Ax}_0+\boldsymbol{\eta}$, where $\boldsymbol{\eta}$ is a vector of outliers with support $T$ and cardinality $|T|= pm$. Assume that the fraction of outliers $p$ satisfies $p<\frac{1}{2}-\frac{\delta_{s}}{1-\delta_{s}}$. The adaptive step size $t_{k,0}$ for the GFHTP$_1$ algorithm is defined as  $\mu_{k,0}\sqrt{\frac{\pi}{2}}\|(\mathbf{b-Ax}^{k-1})\odot (\mathbb{I}_{\{|b_i-(\mathbf{Ax}^{k-1})_i|\leq \theta_\tau(|\mathbf{b-Ax}^{k-1}|)\}})_{i=1}^m\|_1$, where $\mu_{k,0}$ satisfies
\begin{align}\label{GFHTP_condition1}
0<\beta'_{k}:=&1+2\tau^2\left(\Phi^{-1}\left(\frac{1+\tau+p}{2}\right)+\epsilon\right)^2(1+\delta_{s})^2\mu_{k,0}^2 \nonumber \\
&- 2c\sqrt{\frac{2}{\pi}}\left(\tau-\frac{|T_1^{k,0}|}{m}\right)(1-\delta_{s})\mu_{k,0}<\frac{1}{2+\lambda^2},
\end{align}
and $T_1^{k,0}= T\cap \Gamma^{k,0}, \Gamma^{k,0}=\{i:|b_i-(\mathbf{Ax}^{k-1})_i|\leq \theta_\tau(|\mathbf{b-Ax}^{k-1}|)\}$, $c$ is defined as $(2-2p)(1-\delta_s)-(1+\delta_s)$ and $k\leq s+1$. Under these conditions, with a probability of at least $1-n^{-c''}$, the sequence of supports $(S^k)$ and the sequence of estimates $(\mathbf{x}^k)$ generated by GFHTP$_1$ will satisfy the following results at iteration $s$:
\begin{equation*}
S^s=\mathrm{supp}(\mathbf{x}_0), \  \mathbf{x}^s= \mathbf{x}_0.
\end{equation*}
\end{theorem}

\begin{remark}
In fact, the condition \eqref{GFHTP_condition1} of Theorem \ref{propGFHTP} can be met. We consider the following special case: assume that $\lambda= 1$ and $\epsilon= 0.001, \delta_{s}= 0.01, |T_1^{k,l}|/m= 0.001$, $\tau= 0.1:0.001:0.5$, $p= 0.001:0.0001:0.5$. The selectable range for the outliers proportion $p$ is shown in Figure \ref{max_p} (b). Especially, 
\begin{itemize}
\item [\rm{(i)}] when $\tau= 0.4, p= 0.01$, we can take $1.4444<\mu_{k,0}<4.8518$ to make this condition \eqref{GFHTP_condition1} true;
\item [\rm{(ii)}] we can set $\tau= 0.1, p= 0.15$, then we can find that when $7.4256<\mu_{k,0}<43.0710$, the condition \eqref{GFHTP_condition1} is satisfied.
\end{itemize}
\end{remark}

Lastly, we give a remark which provides the comparison with the state-of-the-art methods.
\begin{remark}\label{comparewithAIHT.PSGD.GHTP}
\begin{itemize}
\item [\rm{(i)}] In \cite{Li2023AIHT}, the theoretical findings presented are not directly suitable for scenarios with a significant presence of outliers, given that $\|\boldsymbol{\eta}\|_1$ is substantial. Furthermore, the sparsity level $s$ is frequently unknown in real-world applications. Our work provides a theoretical analysis of the GFHTP$_1$ algorithm, addressing the limitation identified in \cite{Li2023AIHT}. We also notice that Li et al. \cite{Li2023AIHT} choose an adaptive step size $t_k=\mu\|\mathbf{b}-\mathbf{A}\mathbf{x}^k\|_1$ for a fixed constant $\mu$, which is effective primarily for bounded noise. While our truncated adaptive step size $t_{k,l}= \mu_{k,l}\sqrt{\frac{\pi}{2}}\|(\mathbf{b-Au}^{k,l})\odot (\mathbb{I}_{\{|b_i-(\mathbf{Au}^{k,l})_i|\leq \theta_\tau(|\mathbf{b-Au}^{k,l}|)\}})_{i=1}^m\|_1$ is designed to handle outliers more effectively.
\item [\rm{(ii)}] In contrast to the step size selection in the PSGD algorithm proposed by Xu et al. \cite{Xu2024PSGD}, which hinges on the actual signal characteristics, this dependency is impractical. Therefore, our approach determines the step size independently of the signal itself. Furthermore, we also design a high-efficiency stopping criterion  $\|(\mathbf{b-Ax}^{k+1})\odot (\mathbb{I}_{\{|b_i-(\mathbf{Ax}^{k+1})_i|\leq \theta_\tau(|\mathbf{b-Ax}^{k+1}|)\}})_{i=1}^m\|_1\leq \epsilon_{\mathrm{outer}}$, which ensures fast and high-precision convergence, thereby addressing the lack of a suitable stopping criterion in \cite{Xu2024PSGD}.
\end{itemize}
\end{remark}

\section{The Roadmap and Keystone of Our Proofs}\label{sec3roadmap}
\hskip\parindent

In this section, we present the proof strategies for Theorems \ref{thm3.1} and \ref{propGFHTP}. This section serves as the foundation for the subsequent convergence analysis in the Appendices. 

\subsection{The Roadmap and Keystone of the Proof of Theorem \ref{thm3.1}}
\hskip\parindent

The proof strategy of Theorem \ref{thm3.1}  can be summarized in the following two key steps:
\begin{enumerate}
\item[(a)]Contraction of the inner iteration: The central step is to establish the following inequality:
\begin{equation*}
\|\mathbf{x}_0-\mathbf{u}^{k,l+1}\|_2^2
\leq \rho_{k,l}\|\mathbf{x}_0-\mathbf{u}^{k,l}\|_2^2+2\rho_{k,0}\|\mathbf{x}_0-\mathbf{x}^{k-1}\|_2^2,
\end{equation*}
where $\rho_{k,l}= 1+\tau^2(\Phi^{-1}+\epsilon)^2(1+\delta_{2k+s-1})^2\mu_{k,l}^2- 2c_k\sqrt{\frac{2}{\pi}}\left(\tau-\frac{|T_1^{k,l}|}{m}\right)(1-\delta_{2k+s-1})\mu_{k,l}$ with 
 $\Phi^{-1}:=\Phi^{-1}(\frac{1+\tau+p}{2})$, and $T_1^{k,l}= T\cap \Gamma^{k,l}$ with $\Gamma^{k,l}=\{i:|b_i-(\mathbf{Au}^{k,l})_i|\leq \theta_\tau(|\mathbf{b-Au}^{k,l}|)\}$.
\item[(b)]Convergence of the iteration: By induction, we prove 
\begin{align}\label{errorestimation.eq4}
\|\mathbf{x}_0-\mathbf{x}^{k}\|_2^2 = \|\mathbf{x}_0-\mathbf{u}^{k,L+1}\|_2^2
\leq\left(\frac{\rho_{k}^{L+1}(1-3\rho_{k})+2\rho_{k}}{1-\rho_{k}}\right)\|\mathbf{x}_0-\mathbf{x}^{k-1}\|_2^2,
\end{align}
where $\rho_k=\max_{l}\rho_{k,l}$ and $k\geq s$.
\end{enumerate}

Notice that in Step (a), we need to show the contraction property of the inner iteration. To complete the key step, we need two key propositions. The first one gives the upper and lower bounds for the quantity $\|(\mathbf{b-Ax})\odot (\mathbb{I}_{\{|b_i-(\mathbf{Ax})_i|\leq \theta_\tau(|\mathbf{b-Ax}|)\}})_{i=1}^m\|_1$.
\begin{proposition}\label{prop1}(The Sandiwich Inequality for the Quantile Truncation)
For fixed $\epsilon\in (0,1)$, a matrix $\mathbf{A}=[\mathbf{a}_1,\cdots, \mathbf{a}_m]^{\top}\in \mathbb{R}^{m\times n}\ (m\ll n)$ with i.i.d. Gaussian entries, ${a}_{ij}\sim \mathcal{N}(0, \frac{1}{m^2})$. 
\begin{enumerate}
\item[{\rm(i)}] If $m\geq d_0(\epsilon^{-2}\log (\epsilon^{-1}))(s+k)\log n$ for some large enough constant $d_0$, then with probability at least $1-d_1\exp(-d_2m\epsilon^2)-\frac{1}{n}$, where $d_1$ and $d_2$ are some constants, we have for all signals $\mathbf{x}\in \mathbb{R}^n$ with $\|\mathbf{x}\|_0= s+k \ (k\in \mathbb{Z}^{+})$:
\begin{equation}
m\theta_{\tau}(|\mathbf{Ax}|)\in \left[\Phi^{-1}\left(\frac{1+\tau}{2}\right)-\epsilon, \Phi^{-1}\left(\frac{1+\tau}{2}\right)+\epsilon\right]\|\mathbf{x}\|_2.
\end{equation} 
\item[{\rm(ii)}] For the model \eqref{problem1} contaminated by outliers $\boldsymbol{\eta}$, with $T=\mathrm{supp}(\boldsymbol{\eta})$ and $p=\frac{|T|}{m}$, we derive the following lower and upper bounds with high probability:
\begin{align*}\label{upperbound}
&\left(\tau -\frac{|T_1|}{m}\right)\sqrt{\frac{2}{\pi}}(1-\delta_{s+l})\|\mathbf{x-x}_0\|_2\nonumber\\
&\leq \|(\mathbf{b-Ax})\odot (\mathbb{I}_{\{|b_i-(\mathbf{Ax})_i|\leq \theta_\tau(|\mathbf{b-Ax}|)\}})_{i=1}^m\|_1\leq \tau 
\left(\Phi^{-1}\left(\frac{1+\tau+p}{2}\right)+\epsilon\right)\|\mathbf{x-x}_0\|_2,
\end{align*}
where $T_1= T\cap \Gamma$ with $\Gamma=\{i:|b_i-(\mathbf{Ax})_i|\leq \theta_\tau(|\mathbf{b-Ax}|)\}$, $\|\mathbf{x}_0\|_0= s$, and $\|\mathbf{x}\|_0= \ell$.
\end{enumerate}
\end{proposition}

The second proposition implies that the subgradient descent update, $\mathbf{u}^{k,l}+t_{k,l}\mathbf{A}^\top \mathrm{sign}(\mathbf{b}-\mathbf{Au}^{k,l})$, remains close to the true signal  $\mathbf{x}_0$, and that the sequence of errors $\{\mathbf{u}^{k,l}-\mathbf{x}_0\}_{l=1}^{L+1}$ contracts. 
\begin{proposition}\label{pro.ContractionOfInnerIteration}
 (Error Contraction in Inner Iterations) Let $S, S^{k-1}$ and $S^{k}$  denote the supports of $\mathbf{x}_0, \mathbf{x}^{k-1}$ (or $\mathbf{u}^{k,0}$), and $\mathbf{u}^{k,l}$ for $1\leq l\leq L+1$, respectively. 
 \begin{itemize}
 \item[{\rm (i)}] Denote the index set
\begin{equation*}\label{indexset}
\Lambda^{k}:= S\cup S^{k-1}\cup S^{k}.
\end{equation*}
Then we have
\begin{equation}
\|[\mathbf{x}_0-\mathbf{u}^{k,l}-t_{k,l}\mathbf{A}^\top \mathrm{sign}(\mathbf{b}-\mathbf{Au}^{k,l})]_{\Lambda^{k}}\|_2^2
\leq \rho_{k,l}\|\mathbf{x}_0 - \mathbf{u}^{k,l}\|_2^2,
\end{equation}
where $\rho_{k,l}=1+\tau^2\left(\Phi^{-1}\left(\frac{1+\tau+p}{2}\right)+\epsilon\right)^2(1+\delta_{2k+s-1})^2\mu_{k,l}^2- 2c_k\sqrt{\frac{2}{\pi}}\left(\tau-\frac{|T_1^{k,l}|}{m}\right)(1-\delta_{2k+s-1})\mu_{k,l}$.
\item[{\rm(ii)}] Moreover, we have
\begin{align*}\label{errorestimation.eq3.add1}
\|\mathbf{x}_0-\mathbf{u}^{k,l+1}\|_2^2\leq \rho_{k,l}\|\mathbf{x}_0-\mathbf{u}^{k,l}\|_2^2+2\rho_{k,0}\|\mathbf{x}_0-\mathbf{u}^{k,0}\|_2^2.
\end{align*}
\end{itemize}
\end{proposition}

\subsection{The Roadmap and Keystone of the Proof of Theorem \ref{propGFHTP}}
\hskip \parindent

We first list the outline of the proof of Theorem  \ref{propGFHTP}.
Define the set $S$ as the support of $\mathbf{x}_0$ and introduce two random variables, $\zeta_k$ and $\xi_k$, for $k\in [[s]]$, as follows:
\begin{align*}
 \zeta_k:=  & [(\mathbf{x}^{k-1}+t_{k,0}\mathbf{A}^\top\mathrm{sign}(\mathbf{b-Ax}^{k-1}))_S]_k^*,  \\
 \xi_k:=  & [(\mathbf{x}^{k-1}+t_{k,0}\mathbf{A}^\top\mathrm{sign}(\mathbf{b-Ax}^{k-1}))_{S^c}]_1^*.
\end{align*}
Here, $\zeta_k$ represents the $k$-th largest value of the elements in the subset $S$ of $|(\mathbf{x}^{k-1}+t_{k,0}\mathbf{A}^\top\mathrm{sign}(\mathbf{b-Ax}^{k-1}))_j|$, while $\xi_k$ denotes the largest value of the elements in the complementary subset $S^c$. Our outline can be summarized in the following three key steps:
\begin{enumerate}
\item[(a')] Establish that with high probability, $S^k\subseteq S$ for each $k\in [[s]]$. This is implied by the condition $\zeta_k>\xi_k$ for all $k\in [[s]]$. 
\item[(b')] Show that the event $\mathcal{E}=\{(\exists k\in [[s]]: \xi_k\geq \zeta_k \ and \ (\zeta_{k-1}>\xi_{k-1}, \cdots, \zeta_1>\xi_1))\}$ occurs with small probability. In particular, we obtain $\zeta_k > \xi_k$ for all $k\in[[s]]$ with high probability, which implies that $S^s = S$.
\item[(c')] Obtain the conclusion $\mathbf{x}^s= \mathbf{x}_0$, due to the contradiction that $\|\mathbf{x}_0-\mathbf{x}^s\|_2^2
  <\|\mathbf{x}_0-\mathbf{x}^s\|_2^2$.
\end{enumerate}

In order to prove the key step (a'), we need the following key proposition.
\begin{proposition}\label{prop.IndificationSupportSet}
(Support Recovery via Threshold Comparison) The event $\mathcal{F}=\{\xi_k\geq \zeta_k, k\in[[s]]\}$ occurs with probability at most 
\begin{align*}
 2(n-s)\exp\left(-\frac{\gamma_{k}^2m}{6s}\right),
\end{align*}
where $\gamma_k= \frac{\sqrt{\frac{2}{\pi}}}{\mu_{k,0}\tau(\Phi^{-1}(\frac{1+\tau+p}{2})+\epsilon)}\left(\frac{\sqrt{1-2\beta'_{k}}}{\lambda}-\sqrt{\beta_{k}}\right)$, $\beta_k= 1+\tau^2(\Phi^{-1}(\frac{1+\tau+p}{2})+\epsilon)^2(1+\delta_{s})^2\mu_{k,0}^2- 2c\sqrt{\frac{2}{\pi}}\left(\tau-\frac{|T_1^{k,0}|}{m}\right)(1-\delta_{s})\mu_{k,0}$, and $\beta_k'$ is denoted by Theorem \ref{propGFHTP}.  Moreover, one has $S^k\subseteq S$ for each $k\in [[s]]$ with high probability. 
\end{proposition}

\section{Numerical Experiments}\label{s4}
\hskip\parindent

In this section, we embark on numerical experiments to demonstrate the efficacy of the GFHTP$_1$ algorithm in tackling sparse signal recovery challenges, particularly in the presence of outliers. Our synthetic and real-world data experiments validate the robustness and efficiency of GFHTP$_1$ in identifying and mitigating the impact of outliers, aligning with the theoretical guarantees presented in Theorems \ref{thm3.1} and \ref{propGFHTP}.
\subsection{Experiments Settings}\label{s4.1}
\hskip\parindent

We commence by outlining the foundational parameters for our experiments. Firstly, the sensing matrix $\mathbf{A}\in \mathbb{R}^{m\times n}$ is constructed with entries drawn from the i.i.d. Gaussian distribution $\mathcal{N}(0,\frac{1}{m^2})$.  For the purpose of this section, we generate two types of $s$-sparse underlying signal $\mathbf{x}_0\in \mathbb{R}^n$, referred to as 
\begin{itemize}
  \item[(i)] the Gaussian signal with its non-zero entries being i.i.d. standard Gaussian variables;
  \item[(ii)] the `flat' signal with $(\mathbf{x}_0)_j= 1$ for $j\in S$, where $S= \mathrm{supp}(\mathbf{x}_0)$ and $|S|= s$.
\end{itemize} 
To introduce outliers, we first identify $pm$ random positions, where the ratio $p$ represents the proportion of non-zero elements in $\boldsymbol{\eta}$. We consider two types of corruptions:
\begin{itemize}
  \item[(i)] Gaussian outliers: we populate these positions with i.i.d. Gaussian entries with mean zero and variance $\sigma_{\mathrm{outliers}}^2$, while the rest are set to zero;
  \item[(ii)] Uniform outliers: we populate these positions with uniform distribution $\mathcal{U}(-u_{\mathrm{outliers}}, u_{\mathrm{outliers}})$, while the rest are set to zero.
\end{itemize} 

The measurement vector
$\mathbf{b}$ is constructed in accordance with Equation \eqref{problem1}, with $b_i= \mathbf{a}_i^\top \mathbf{x}_0+\boldsymbol{\eta}_i$ for $i\in T$ and $b_i= \mathbf{a}_i^\top \mathbf{x}_0$ for $i\in T^c$, where $T= \mathrm{supp}(\boldsymbol{\eta})$ denotes the support of the outliers $\boldsymbol{\eta}$. 

Lastly, we mark $\hat{\mathbf{x}}$ as the recovered signal. We evaluate the numerical performance using the following metric for all algorithms:
\begin{itemize}
  \item[(i)] Relative error ($\mathrm{RelErr}$):  $\mathrm{RelErr}(\hat{\mathbf{x}}, \mathbf{x}_0)= \|\hat{\mathbf{x}}-\mathbf{x}_0\|_2/\|\mathbf{x}_0\|_2$;
  \item[(ii)] Success rate (SR): the success rate of recovery in 100 trials, and a successful reconstruction is declared when $\mathrm{RelErr}(\hat{\mathbf{x}}, \mathbf{x}_0)\leq \epsilon $ with $\epsilon=10^{-4}$.
\end{itemize}
All settings, unless otherwise specified, shall be selected as above. Each experiment is replicated $100$ times to ensure statistical significance and all results report the mean value over 100 independent trials. Throughout Subsections \ref{s4.2} and \ref{s4.3}, we set the signal and measurement dimensions to $n= 5000$ and $m= 1000$.

\subsection{The Performance of Our Algorithms}\label{s4.2}
\hskip\parindent
{\bf a) The Chooses of Step Sizes $\mu_{k,l}$}

Firstly, we test the performance of different step sizes $\mu_{k,l}$ in Remark \ref{remark2.3}. Here we take the gaussian outliers with the parameter $\sigma_{\mathrm{outliers}}= 10$, the quantile $\tau= 0.5$, and $s= 5, 10, 15$. The numbers of iterations for the outer and inner iterations are $\mathrm{MaxIt}= 30$ and $L= 10$ respectively. Figure \ref{u_choose} shows the performance of the algorithm in terms of relative error, average CPU time, and the sparsity of the recovered signal when the sparsity is $5, 10,$ and $15$, the outliers ratio ranges from $0.05$ to $0.5$, and different $\mu_{k,l}$ values are selected. Notably, our theoretical result requires that $\mu_{k,l}$  to be fixed within the interval $(1.3695, 3.3362)$ when $\tau= 0.5, p= 0.05$. However from Figure \ref{u_choose}, we can see that our GFHTP$_1$ exhibits superior performance when the step size is $\mu_{k,l}= 4, 6, 8$. Especially, when $\mu_{k,l}= 1$, the relative error is small, however the sparsity of the recovered signal is not exact. This discrepancy between the theoretical result and empirical effectiveness highlights an important direction for future work. Motivated by these experimental findings, we adopt $\mu_{k,l}= 6$ for all subsequent experiments.
\begin{figure}[htbp]
\centering
\subfigure[The performance when $s= 5$]
{
\begin{minipage}{15cm}
\centering
\includegraphics[width=\linewidth]{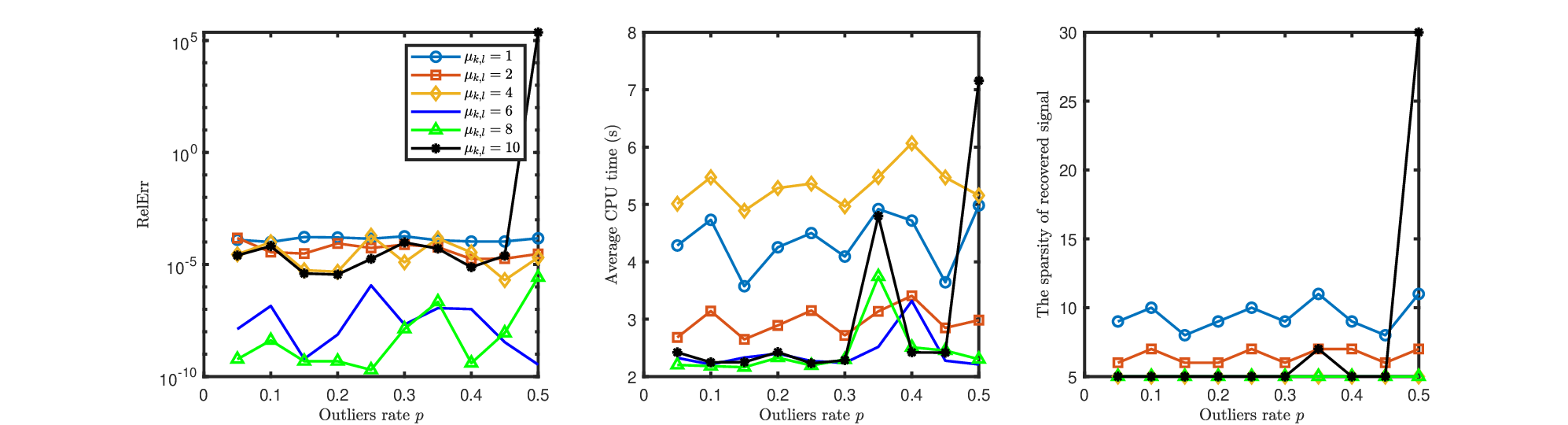}%
\end{minipage}
}
\subfigure[The performance when $s= 10$]
{
\begin{minipage}{15cm}
\centering
\includegraphics[width=\linewidth]{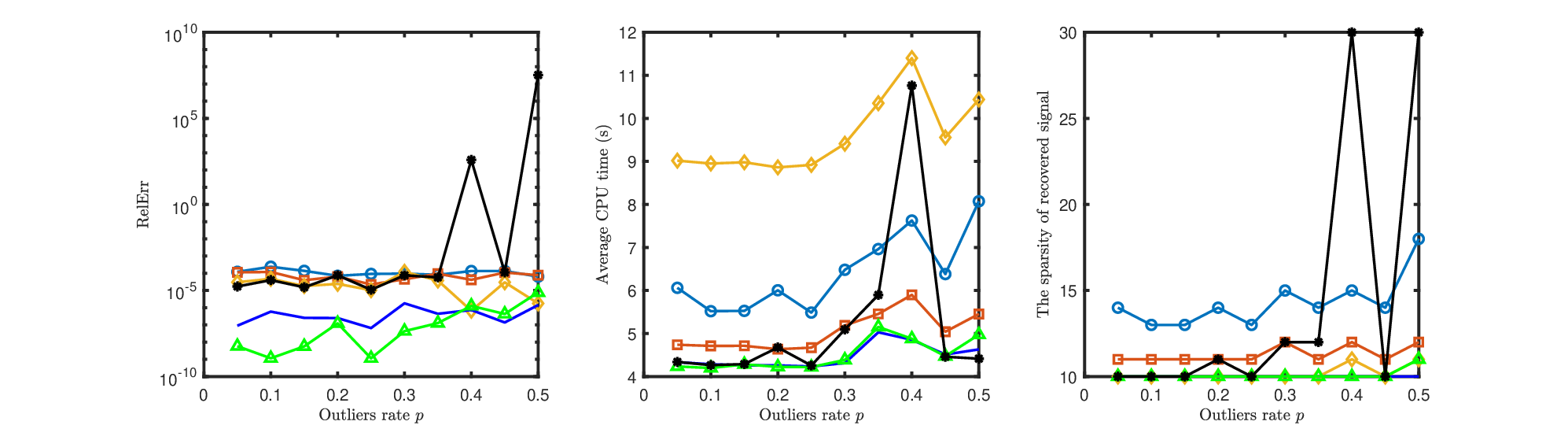}%
\end{minipage}
}
\subfigure[The performance when $s= 15$]
{
\begin{minipage}{15cm}
\centering
\includegraphics[ width=\linewidth]{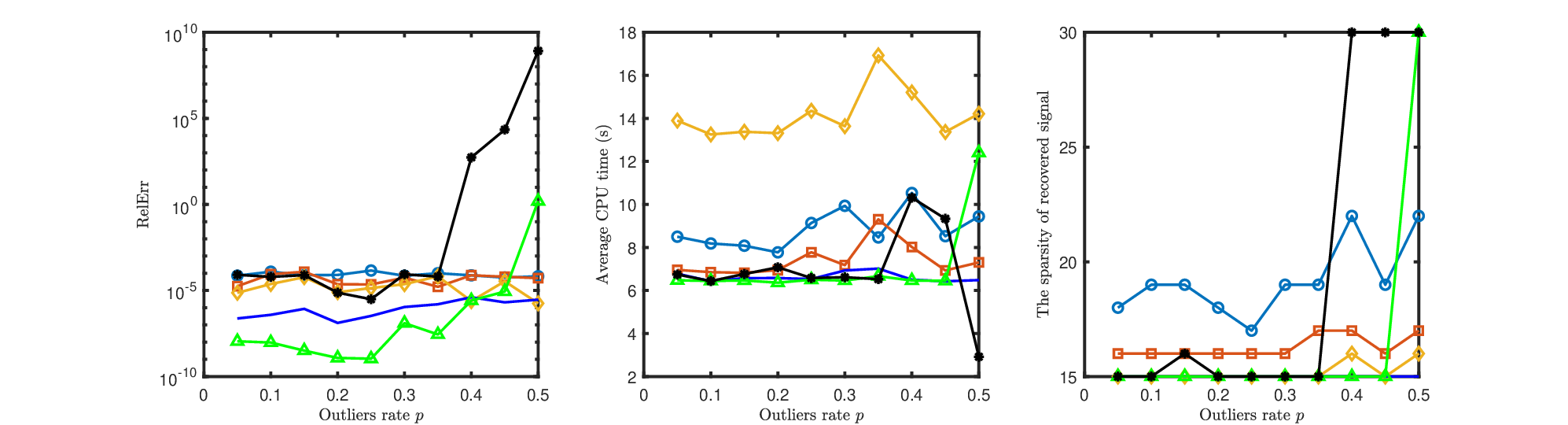}%
\end{minipage}
}
\caption{Relative error, average CPU time, the sparsity of the recovered signal for the GFHTP$_1$ algorithm using Gaussian outliers.}
\label{u_choose}
\end{figure}

{\bf b) The Choices of Inner Iterations $L$}

Next, we explore the performance of different inner iterations $L$. Some fundamental settings are $\sigma_{\mathrm{outliers}}= 10$, the quantile $\tau= 0.5$, $\mu_{k,l}= 6$, and $s= 5, 10, 15$. The number of iterations for the outer iteration is $\mathrm{MaxIt}= 30$. We show this result in Figure \ref{max_iter}. From Figure \ref{max_iter}, we can know that as $L$ increases, the relative error decreases, but the time increases accordingly. Meanwhile, the sparsity of recovered signal is inaccurate when $L= 1$. Therefore, considering both time and accuracy, we choose $L= 10$.
\begin{figure}[htbp]
\centering
\subfigure[The performance when $s= 5$]
{
\begin{minipage}{15cm}
\centering
\includegraphics[width=\linewidth]{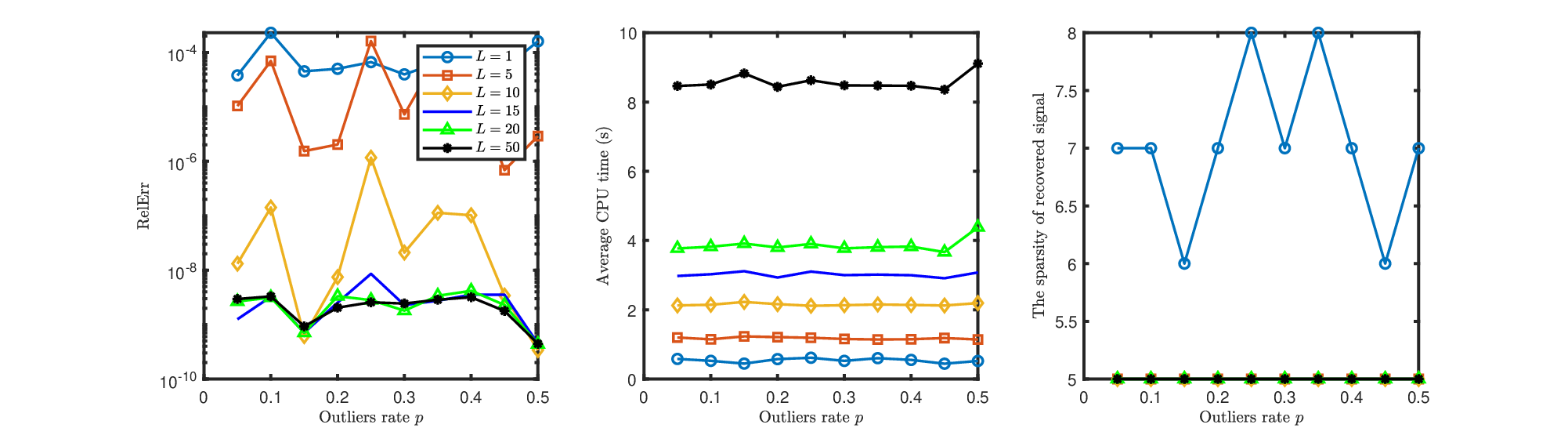}%
\end{minipage}
}
\subfigure[The performance when $s= 10$]
{
\begin{minipage}{15cm}
\centering
\includegraphics[width=\linewidth]{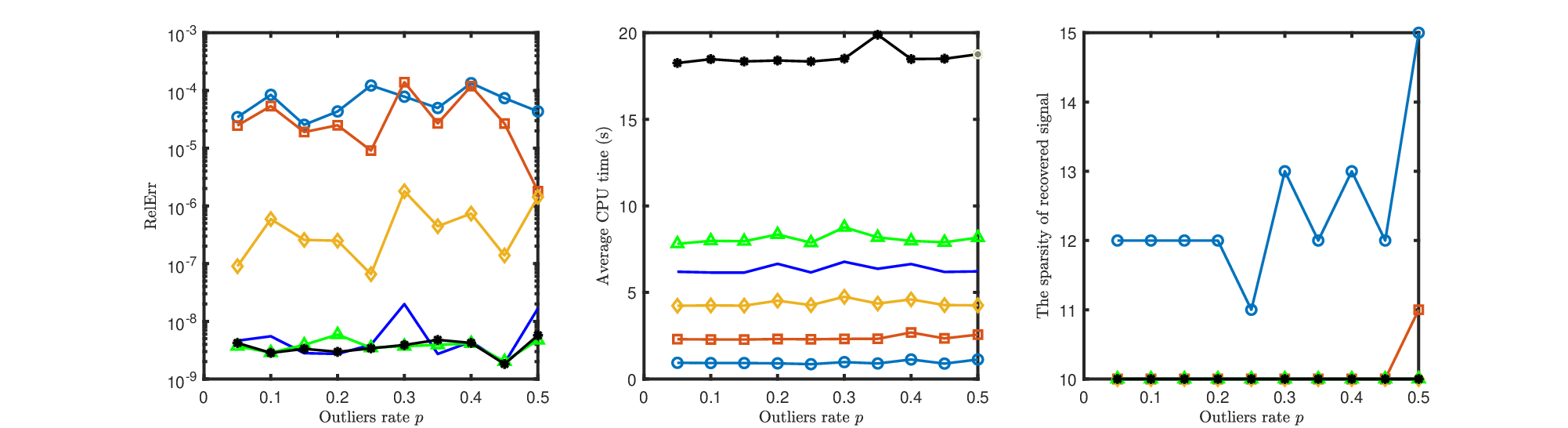}%
\end{minipage}
}
\subfigure[The performance when $s= 15$]
{
\begin{minipage}{15cm}
\centering
\includegraphics[ width=\linewidth]{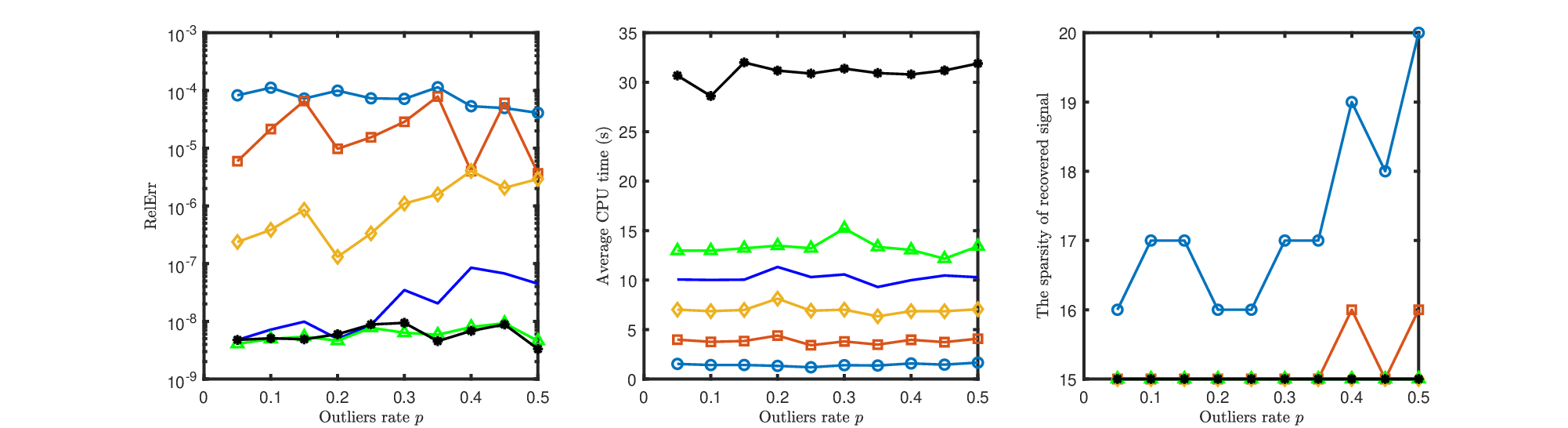}%
\end{minipage}
}
\caption{Numerical performance for the GFHTP$_1$ algorithm using Gaussian outliers when various $L$.}
\label{max_iter}
\end{figure}

{\bf c) Performance of Different Quantiles $\tau$}

Then, we consider the result of different quantiles $\tau$. Some fundamental settings are $\sigma_{\mathrm{outliers}}= 10, u_{\mathrm{outliers}}= 100$, $s= 5$, $\mu_{k,l}= 6$, $L= 10$. The number of iterations for the outer iteration is $\mathrm{MaxIt}= 30$. Figure \ref{tau_choose} shows the performance of the algorithm in terms of relative error when the sparsity is $5$, the outliers ratio ranges from $0.05$ to $0.55$, and different $\tau$ values are selected. From Figure \ref{tau_choose}, when the quantile $\tau$ is $0.35$, it can be seen that a more accurate recovery can be achieved even when the proportion of outliers is $0.55$. When the quantile $\tau$ is $0.5$, the effect is very good except when the proportion of outliers is $0.55$. However, when the quantile $\tau$ is $0.8$, the effect is not satisfactory even when the proportion of outliers is $0.2$, that is, $\tau \leq 1-p$. Therefore, in the subsequent experiments, we choose a more general quantile, namely the median.
\begin{figure}
\centering
\includegraphics[width=\linewidth]{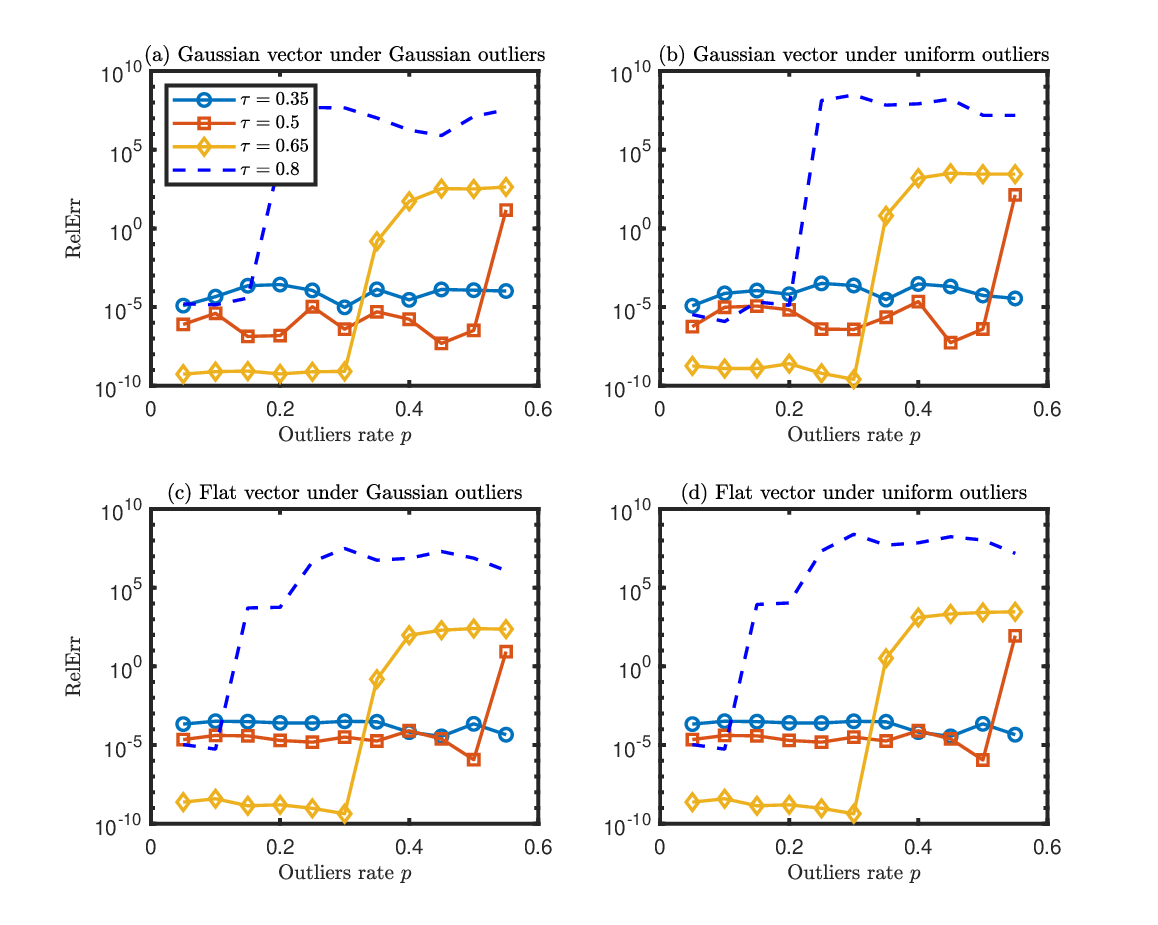}
\caption{Numerical performance for the GFHTP$_1$ algorithm when various $\tau$.}
\label{tau_choose}
\end{figure}

Similarly, we can get same results using the FHTP$_1$ algorithm. Thus, we set $\mu_{k,l}= 6, L= 10$. 

Then, our result is shown in Figure \ref{ourperformance}, thus implying that the GFHTP$_1$ with adaptive step sizes can exactly recover the underlying sparse signal $\mathbf{x}_0$ when the sparsity of the real signal is unknown. 
\begin{figure}[htbp]
\centering
\subfigure[Original and recovered signals]
{
\begin{minipage}{7cm}
\centering
\includegraphics[scale = 0.45]{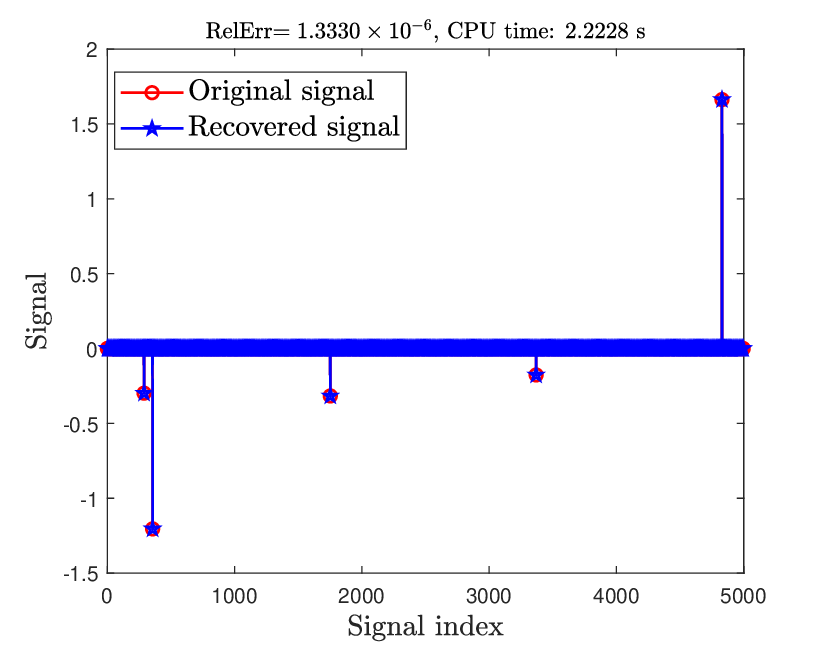}%
\end{minipage}
}
\subfigure[Convergence behaviors of GFHTP$_1$]
{
\begin{minipage}{7cm}
\centering
\includegraphics[scale = 0.45]{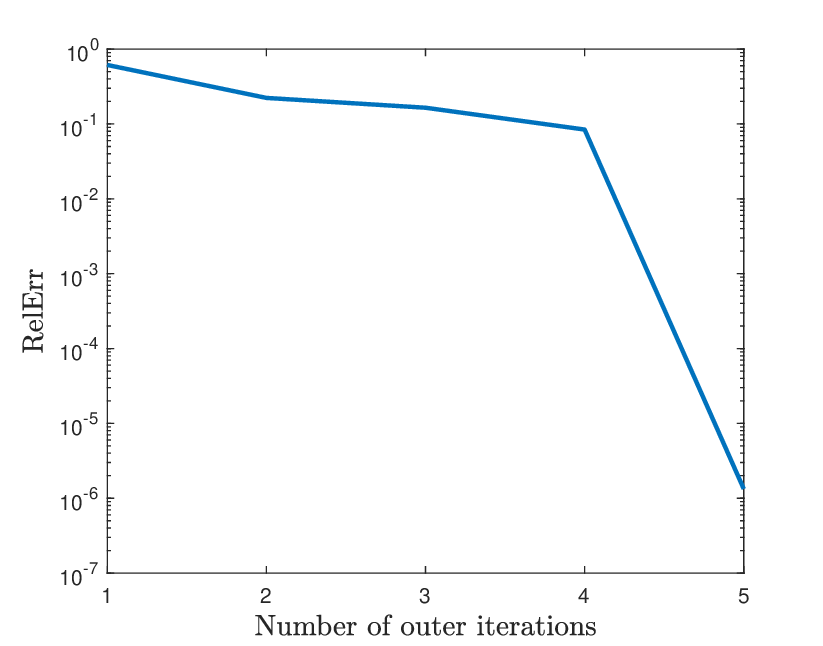}%
\end{minipage}
}
\caption{Performance of our proposed GFHTP$_1$ algorithm for Gaussian vector and Gaussian outliers with $\sigma_{\mathrm{outliers}}= 10$ and $p= 0.2$.}
\label{ourperformance}
\end{figure}
Especially, we assume that the sparsity of the original signal is known, we can show the performance of FHTP$_1$ algorithm in Figure \ref{FHTPperformance}. Compared with Figure \ref{ourperformance}, we can see that FHTP$_1$ is faster  with comparable result than GFHTP$_1$. However, the sparsity is usually unknown and inaccurate sparsity can have a huge impact on the recovery effect. Therefore, studying GFHTP$_1$ makes sense, although it takes more time.
\begin{figure}[htbp]
\centering
\subfigure[Original and recovered signals]
{
\begin{minipage}{7cm}
\centering
\includegraphics[scale = 0.45]{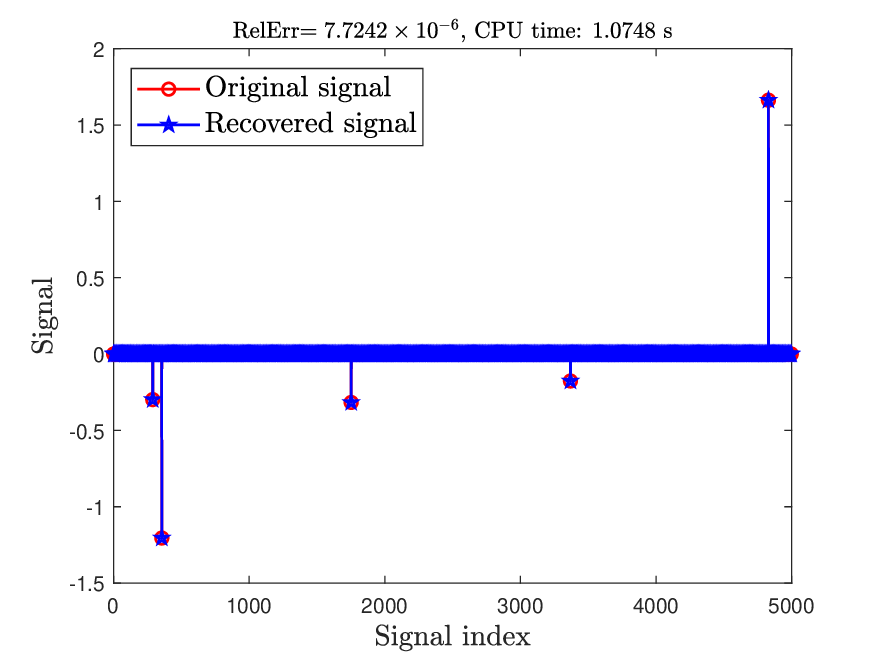}%
\end{minipage}
}
\subfigure[Convergence behaviors of FHTP$_1$]
{
\begin{minipage}{7cm}
\centering
\includegraphics[scale = 0.45]{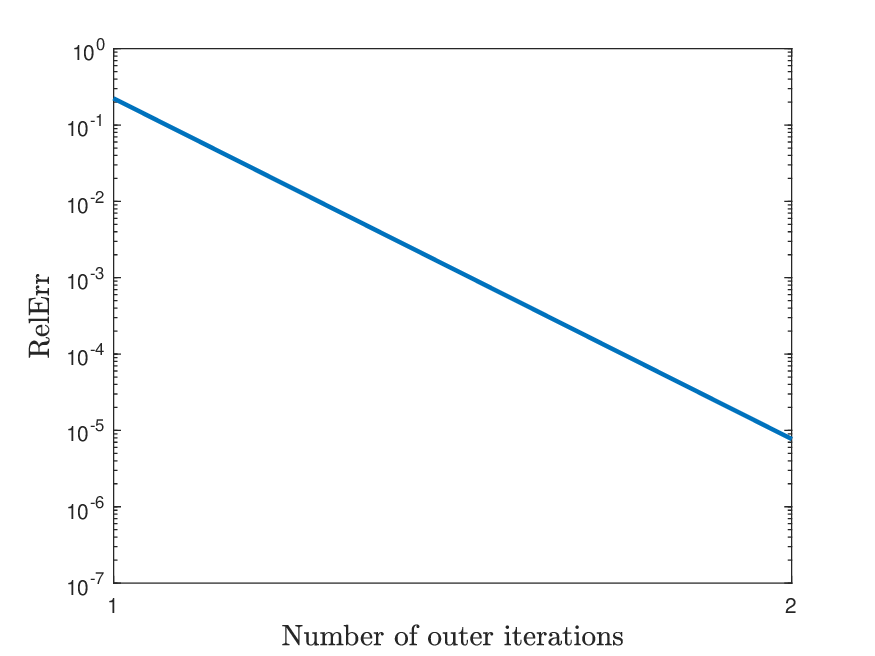}%
\end{minipage}
}
\caption{ Performance of our proposed FHTP$_1$ algorithm for Gaussian vector and Gaussian outliers with $\sigma_{\mathrm{outliers}}= 10$ and $p= 0.2$.}
\label{FHTPperformance}
\end{figure}

Lastly, we verify the effect of Theorem \ref{propGFHTP}. We choose the `flat' vector as original signal. Some fundamental settings are $\sigma_{\mathrm{outliers}}= 10$, $s= 5$, $\mu_{k,l}= 6$, $L= 10$. The result is shown in Figure \ref{GFHTP_flat}. From Figure \ref{GFHTP_flat}, we can see that $S^1, S^2, S^3, S^4, S^5\subseteq S$ and $S^5= \mathrm{supp}(\mathbf{x}_0)$. It can be approximately regarded as $\mathbf{x}^5 = \mathbf{x}_0$. Therefore, Theorem \ref{propGFHTP} has been verified.
\begin{figure}
\centering
\includegraphics[scale = 0.48]{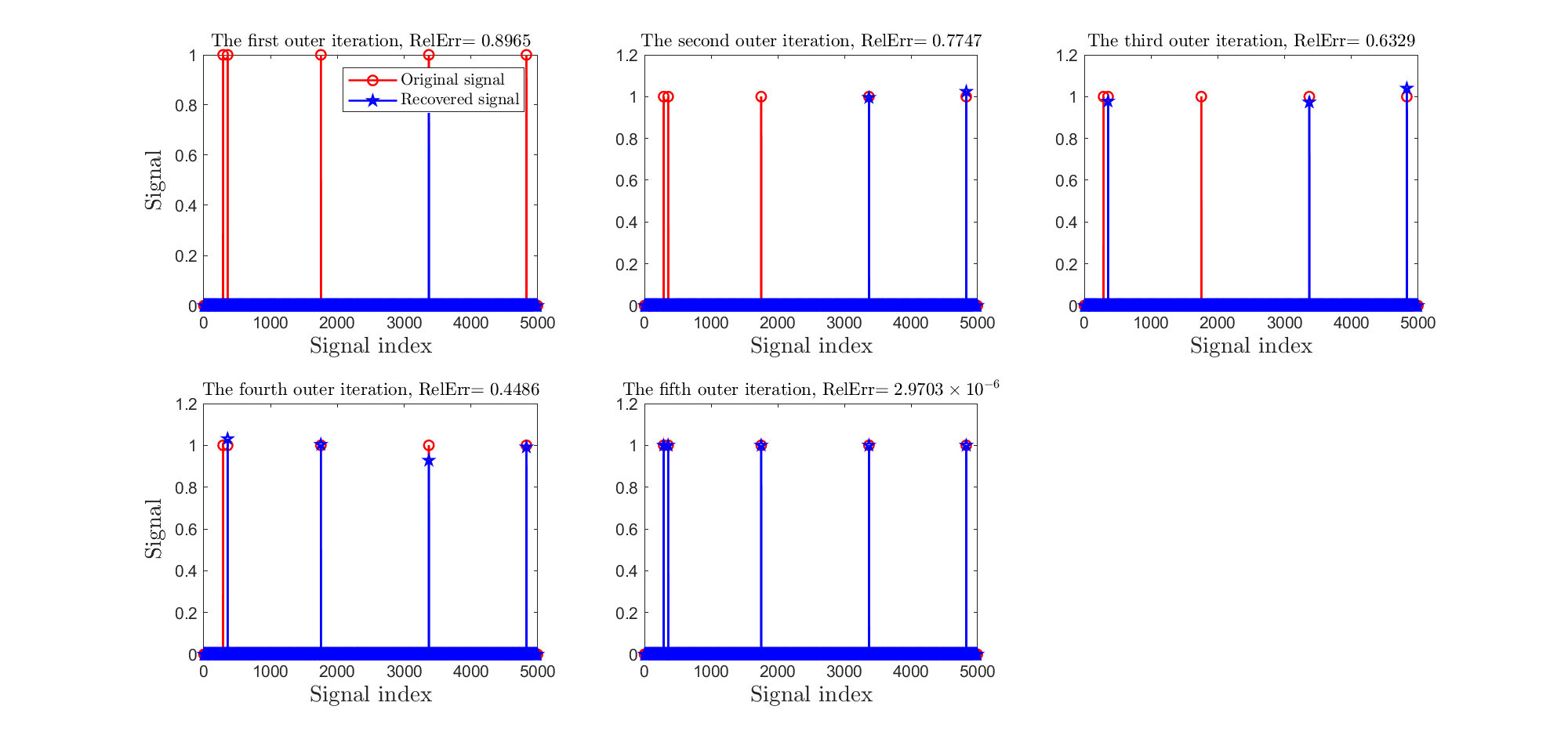}
\caption{Numerical performance for the GFHTP$_1$ algorithm under the `flat' vector case.}
\label{GFHTP_flat}
\end{figure}

\subsection{Comparison with Other Algorithms}\label{s4.3}
\hskip\parindent

In this section, we first compare the several methods with Gaussian vector whose $s$ nonzero entries are independent standard normal random variables in no outliers case. We set $p= 0, L= 10$. We give the fundamental settings of the other existing algorithms in sparse signal recovery: PSGD (with real sparsity $s$), PSGD$^{2}$ (with the sparsity level as $2s$), AIHT (with real sparsity $s$), AIHT$^{2}$ (with the sparsity level as $2s$). The PSGD and PSGD$^{2}$ algorithms utilize the step size $\mu_k=0.8\times 0.95^k$. And we used the criterion $\mathrm{RelErr}(\mathbf{x}^{k+1}, \mathbf{x}^k)\leq 10^{-8}$ or the maximum iteration number $1000$ for PSGD, PSGD$^{2}$, AIHT and AIHT$^{2}$.
We present here a comparison between PSGD, PSGD$^{2}$, AIHT, AIHT$^{2}$ and GFHTP$_1$, FHTP$_1$ in terms of successful recovery. The result is shown in Figure \ref{comparenoiselessnooutliers}. From Figure \ref{comparenoiselessnooutliers}, we can find that our algorithm GFHTP$_1$  displays more robustness on the sparsity than the other algorithms. Our algorithm needs more CPU time than that of AIHT because it takes too much time to search for true sparsity. 
\begin{figure}[htbp]
\centering
\subfigure[Success rates]
{
\begin{minipage}{7cm}
\centering
\includegraphics[scale = 0.45]{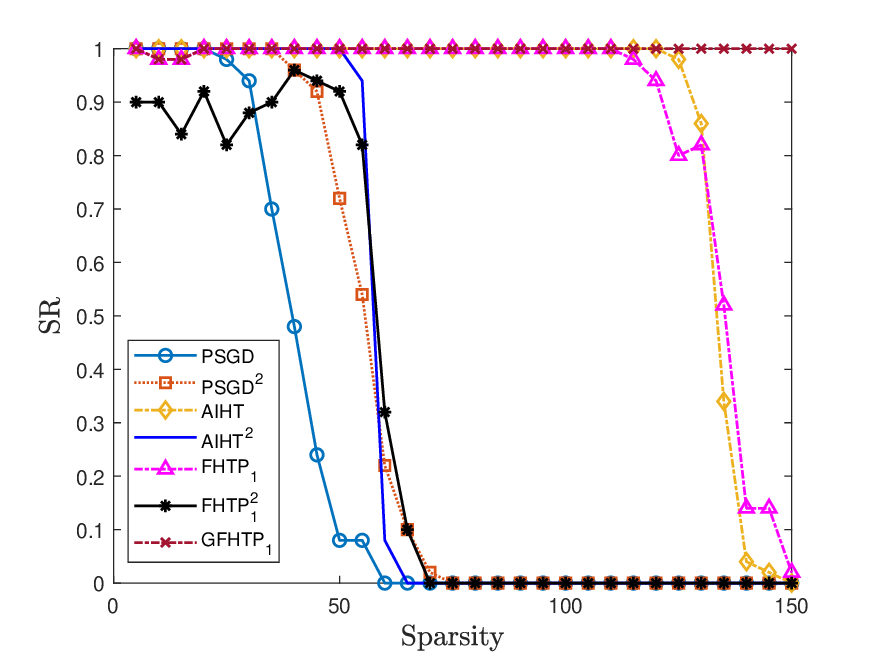}%
\end{minipage}
}
\subfigure[Average CPU time]
{
\begin{minipage}{7cm}
\centering
\includegraphics[scale = 0.45]{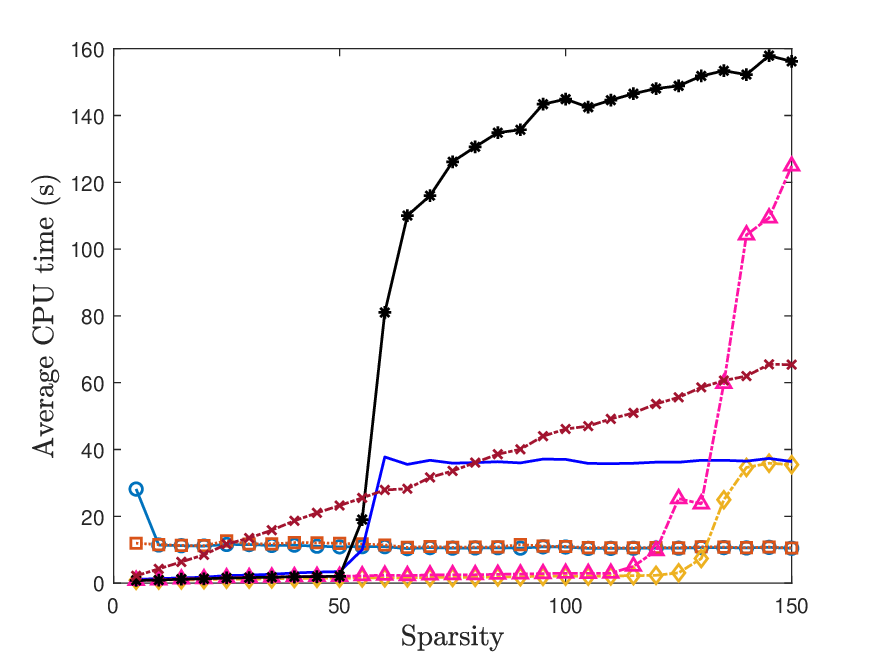}%
\end{minipage}
}
\caption{Rates of successful recoveries and average CPU times for PSGD, PSGD$^2$, AIHT, AIHT$^{2}$ and GFHTP$_1$, FHTP$_1$ using Gaussian measurements.}
\label{comparenoiselessnooutliers}
\end{figure}

Next, we present the results for outliers case. For AIHT algorithm, we set $p= 0.05, 0.1, \sigma_{\mathrm{outliers}}= 10, 100$, the results are shown in Figure \ref{AIHT_two}. It is easy to see that the two methods are invalid for outliers case. So we compare the performance of the remaining algorithms. We set $s= 5, 10, \sigma_{\mathrm{outliers}}= 10$. For different outliers rate $p= 0.05i (i= 1, \cdots, 10)$, the success rates and avrage CPU time are shown in Tables \ref{outliersrate} and \ref{outliersrate_ones}. We can see that our algorithm is also effective when the outliers rate is high. And our algorithm takes less time and works better than PSGD.
\begin{figure}[htbp]
\centering
\subfigure[SNR in different outliers rates and $\sigma_{\mathrm{outliers}}$ under Gaussian vector]
{
\begin{minipage}{7cm}
\centering
\includegraphics[scale = 0.45]{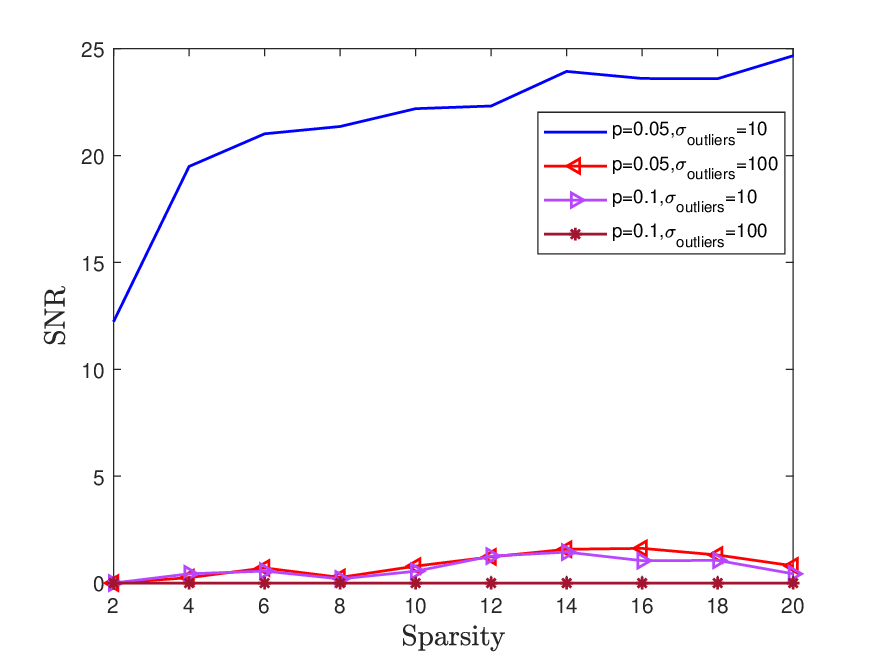}%
\end{minipage}
}
\subfigure[SNR in different outliers rates and $\sigma_{\mathrm{outliers}}$ under `flat' vector]
{
\begin{minipage}{7cm}
\centering
\includegraphics[scale = 0.45]{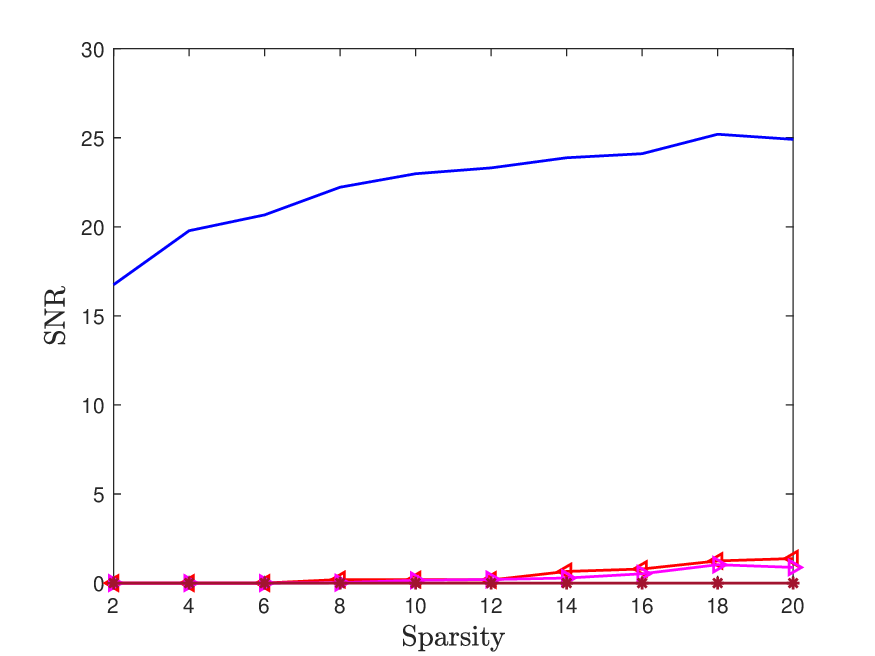}%
\end{minipage}
}
\caption{SNR for AIHT using Gaussian vector and `flat' vector in outliers case.}
\label{AIHT_two}
\end{figure}

\begin{table}[htbp]
    \centering
    \caption{The success rates and average CPU time (s) in different outliers rates and different sparsity by different algorithms under Gaussian vector.}
    \label{outliersrate}
    \setlength{\tabcolsep}{2mm}{
    \begin{tabular}{c|ccc|ccc}
    \hline
         Sparsity& & $s= 5$ & & & $s= 10$ & \\ \hline
        Outliers rates& PSGD & FHTP$_1$ & GFHTP$_1$ & PSGD & FHTP$_1$ & GFHTP$_1$\\ \hline
        $0.05$ & 1,12.2750& 1,0.8126& 1,2.2274& 1,11.9584& 0.99,1.0992& 0.99,4.5080 \\
        $0.10$ &  1,11.9099& 1,0.8284& 1,2.1784& 1,11.9567& 1,1.1192& 1,4.4967 \\
        $0.15$ & 1,11.8507& 1,0.8209& 1,2.1786& 1,11.9303& 0.99,1.1094& 0.99,4.4813 \\
        $0.20$ &  1,11.8399& 1,0.8279& 1,2.1924& 0.98,11.8560& 1,1.1960& 1,4.5566 \\
        $0.25$ & 1,11.8435& 1,0.8766& 1,2.1838& 0.93,11.7049& 0.99,1.1840& 0.99,4.5112 \\
        $0.30$ &  1,11.8230& 1,0.8913& 1,2.1888&  0.86,11.2739& 1,1.2191& 1,4.4501 \\
        $0.35$ & 0.98,11.5808& 1,0.9036& 1,2.1803&  0.76,11.0234& 1,1.3530& 1,4.4405 \\
        $0.40$ &   0.96,11.5065& 1,0.9359& 1,2.1820&  0.47,10.3424&  1,1.3996&  1,4.4441 \\
        $0.45$ &  0.72,10.8374& 1,1.0428&  1,2.1862&   0.25,9.8564&  1,1.4571&  1,4.4539 \\
        $0.50$ &   0.71,10.7115&  1,1.0163&  1,2.1792&   0.05,9.6208&  1,1.5884&  1,4.4620 \\
    \hline
    \end{tabular}}
\end{table}

\begin{table}[htbp]
    \centering
    \caption{The success rates and average CPU time (s) in different outliers rates and different sparsity by different algorithms under `flat' vector.}
    \label{outliersrate_ones}
    \setlength{\tabcolsep}{2mm}{
    \begin{tabular}{c|ccc|ccc}
    \hline
         Sparsity& & $s= 5$ & & & $s= 10$ & \\ \hline
        Outliers rates& PSGD & FHTP$_1$ & GFHTP$_1$ & PSGD & FHTP$_1$ & GFHTP$_1$\\ \hline
        $0.05$ & 1,12.0905& 1,0.4454& 1,2.2258& 1,12.5671& 1,~0.4764& 0.99,4.7804 \\
        $0.10$ &  1,11.4555& 1,0.4230& 1,2.1067& 1,11.9900& 1,~0.4513& 1,4.5219 \\
        $0.15$ & 1,11.4388& 1,0.4247& 1,2.1157& 1,11.4821& 1,~0.4349& 1,4.3494 \\
        $0.20$ &  1,11.3687& 1,0.4214& 1,2.1105& 1,11.4300& 1,~0.4820& 1,4.3460 \\
        $0.25$ & 1,11.3648& 1,0.4233& 1,2.1141& 1,11.3538& 0.99,~0.4893& 1,4.3378 \\
        $0.30$ &  1,11.2867& 1,0.4222& 1,2.1093&  0.91,11.0272& 1,~0.6082& 1,4.3319 \\
        $0.35$ & 1,11.2749& 1,0.4233& 1,2.1129&  0.80,10.7377& 1,~0.6797& 1,4.3307 \\
        $0.40$ &   1,11.2021& 1,0.4270& 1,2.1134&  0.26,9.3547&  1,~0.8074&  1,4.3360 \\
        $0.45$ &  0.94,12.0457& 1,0.5136&  1,2.3220&   0.09,9.2260&  1,~0.8541&  1,4.3486 \\
        $0.50$ &   0.82,10.9066&  1,0.5122&  1,2.1932&   0,9.2853&  1,~0.9455&  1,4.3274 \\
    \hline
    \end{tabular}}
\end{table}

\subsection{Test on Real Data}\label{Test.RealData}
\hskip\parindent

 In this subsection, we validate the effectiveness of our proposed algorithms in the context of image restoration.
 Specially, the MNIST test dataset-comprising 10000 handwritten digit images-is adopted for experimental evaluation. Each image in this dataset has a resolution of $28 \times 28$ pixels, with a black background and white foreground content. Owing to the limited proportion of white areas, the images can be considered sparse. Each image is vectorized into a column vector as the sparse vector $\mathbf{x}_0\in\mathbb{R}^{n}$ with $n=784$. The observation vector $\mathbf{b}$ is obtained through $\mathbf{b}=\mathbf{A}\mathbf{x}_0+\boldsymbol{\eta}$ with the outliers $\boldsymbol{\eta}$. Subsequently, the different algorithms are applied to reconstruct the original sparse 
 signal $\mathbf{x}_0$ with given $\mathbf{b}$ and $\mathbf{A}$. Here $m= 700, p= 0.1, \sigma_{\mathrm{outliers}}= 10$. Figure \ref{Number_image} gives examples of image recovery for different MNIST images. From Figure \ref{Number_image}, we can see that our algorithms outperform the PSGD algorithm. Table \ref{figure_number} presents that our algorithms perform better than the PSGD algorithm, both in terms of SNR and CPU time.
\begin{figure}[htbp]
\centering
\includegraphics[scale = 1]{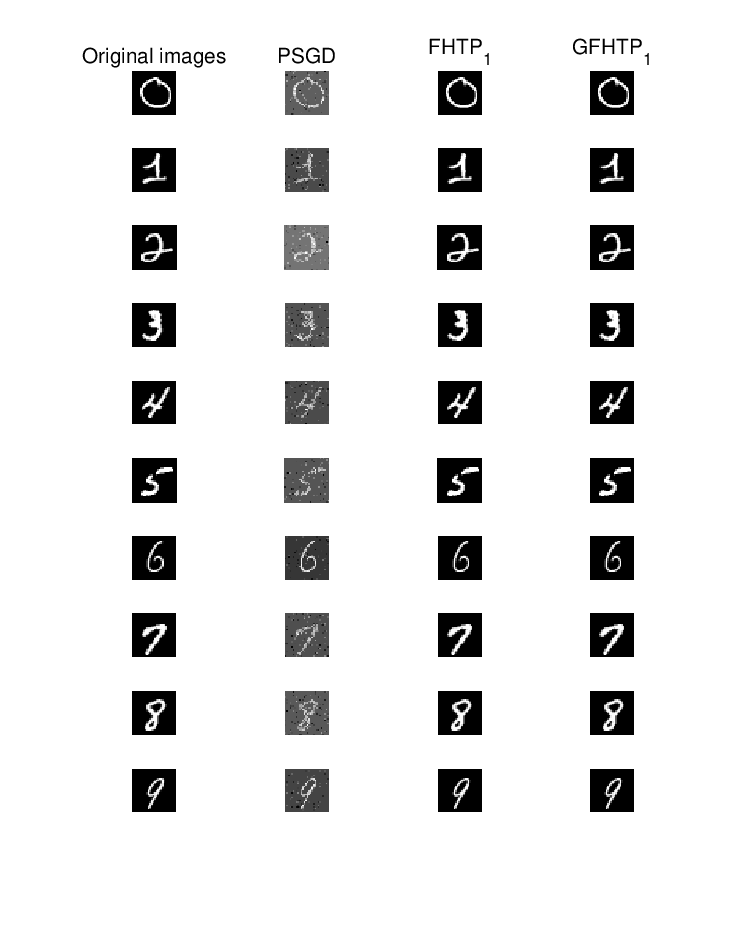}%
\caption{Examples of image recovery for different MNIST dataset images.}
\label{Number_image}
\end{figure}

\begin{table}[htbp]
    \centering
    \caption{The SNR and average CPU time (s) by different algorithms for MNIST dataset images.}
    \label{figure_number}
    \setlength{\tabcolsep}{2mm}{
    \begin{tabular}{c|ccc}
    \hline
        Original Images & PSGD & FHTP$_1$ & GFHTP$_1$ \\ \hline
        Number 0 ($s=138$) & 3.9810,1.3670& \textbf{88.7157},\textbf{0.6711}& 85.4613,9.0405 \\
        Number 1 ($s=139$) & 4.3541,1.4701& 89.8683,\textbf{0.7665}& \textbf{90.2763},9.4709 \\
        Number 2 ($s=150$) & 3.4944,1.3526& \textbf{97.4279},\textbf{0.8991}& 96.2120,10.5938 \\
        Number 3 ($s=155$) & 3.1696,1.3157& \textbf{102.8420},\textbf{1.0409}& 90.5717,9.8892 \\
        Number 4 ($s=130$) & 3.5962,1.3465& \textbf{111.0775},\textbf{0.9729}& 110.1130,8.3666 \\
        Number 5 ($s=111$) & 2.8171,1.2974& 93.1071,\textbf{0.5612}& \textbf{93.4327},7.1790 \\
        Number 6 ($s=107$) & 8.6223,1.3339& \textbf{87.3756},\textbf{0.5871}& 86.6374,7.0019 \\
        Number 7 ($s=144$) & 2.9630,1.3464& \textbf{105.1291},\textbf{0.9642}& 84.1520,9.2865 \\
        Number 8 ($s=158$) & 3.2335,1.3159& 89.6628,\textbf{1.2815}& \textbf{96.4844},10.3156 \\
        Number 9 ($s=108$) & 5.4890,1.3764& 102.5641,\textbf{0.5188}& \textbf{103.1612},7.0324 \\
    \hline
    \end{tabular}}
\end{table}

\section{Conclusions}\label{s5}
\hskip\parindent

This work tackled the critical challenge of sparse signal recovery from measurements corrupted by a constant fraction of gross outliers. To address this issue, we developed a sparsity-constrained least absolute deviations (LAD) minimization model and proposed a novel Graded Fast Hard Thresholding Pursuit algorithm (GFHTP$_1$). Distinguished from existing methods, GFHTP$_1$ has two core advantages: it eliminates the requirement for prior knowledge of signal sparsity, and its adaptive quantile-truncated step size derived from the residual $\mathbf{b} -\mathbf{A}\mathbf{x}^k$ effectively suppresses outlier interference.

For theoretical rigor, we established a novel sandwich inequality for quantile truncation, which provides a rigorous analytical framework for quantifying outlier removal effects. On this basis, we proved a critical proposition for exact support recovery, and derived the key convergence guarantee: an $s$-sparse signal can be exactly recovered within at most $s$ iterations.

Extensive numerical simulations demonstrated that GFHTP$_1$ outperforms state-of-the-art algorithms, especially in challenging scenarios with high outlier ratios and high signal sparsity. Meanwhile, its low CPU time consumption ensured efficiency and scalability for large-scale problems. This work not only advanced the theoretical foundation of sparse recovery under outlier contamination but also offered a practical solution for applications such as wireless sensor networks, image restoration, and compressed sensing.

\section*{Acknowledgements}\label{sec6}
\hskip\parindent

We would like to express our gratitude to Dr. Li-Ping Yin for her assistance in proving Theorem \ref{propGFHTP}. This paper is supported by the National Natural Science Foundation of China (Nos. 12471353, 12201268), Fundamental Research Funds for the Central Universities (No. lzujbky-2024-it51).

\section*{Data Availability} 
\hskip\parindent

The data that support the findings of this study are available from the corresponding author
upon reasonable request.

\section*{Declarations}
\hskip\parindent

{\bf Conflict of Interest} The authors declare that they have no conflict of interest.
\appendix
\section*{Appendices}

\section{Proof for General Sparse Signals}\label{secappendixa}
\subsection{Preliminaries}\label{s2.1}
\hskip\parindent

In this subsection, we establish the foundational concepts and lemmas, including the restricted isometry property, the concentration property of sample quantile, and some necessary inequalities.

\subsubsection{Restricted 1-Isometry Property (RIP$_1$)}\label{s2.1.1}
\hskip\parindent

We introduce a matrix that fulfills the RIP$_1$.
\begin{lemma}\label{RIP}(\cite[Lemma 3.2]{Chartrand2008RIP})
Suppose that $m\geq c_1s\log(n/s)$ for some constant $c_1$. Let $\mathbf{A}\in \mathbb{R}^{m\times n} (m\ll n)$ be a matrix with entries that are i.i.d. Gaussian random variables, i.e., ${a}_{ij}\sim \mathcal{N}(0, \frac{1}{m^2})$. Then, $\mathbf{A}$ satisfies the RIP$_1$ of order $s$ with a high probability exceeding $1-2\exp(-\frac{m\delta_s^2}{16})$.
\end{lemma}

We also present a lemma that is instrumental in proving the main result in the subsequent section.
\begin{lemma}\label{RIP1}
The matrix $\mathbf{A}$ is assumed to satisfy the RIP$_1$, as defined in Lemma \ref{RIP}. Let $W\subseteq [[m]]$ with $q= \frac{|W|}{m}$. Then, the submatrix $\mathbf{A}_{W}\in \mathbb{R}^{|W|\times n}$ of $\mathbf{A}$ also satisfies the RIP$_1$, i.e., we have
\begin{equation}\label{RIP.eq2}
(1-\delta_s)\|\mathbf{x}\|_2\leq \frac{1}{q}\sqrt{\frac{\pi}{2}}\|\mathbf{A}_{W}\mathbf{x}\|_1\leq (1+\delta_s)\|\mathbf{x}\|_2
\end{equation}
for any $s$-sparse vector $\mathbf{x}$ with high probability.
\end{lemma}
\begin{proof}
Since $a_{ij}\sim \mathcal{N}(0, \frac{1}{m^2})$,  scaling by $\frac{1}{q}$ yields $\frac{1}{q}a_{ij}\sim \mathcal{N}(0, \frac{1}{(qm)^2})$. Noting that $qm=|W|$, this simplifies to $\frac{1}{q}a_{ij}\sim \mathcal{N}(0, \frac{1}{|W|^2})$. Consequently, the matrix $\mathbf{A}_{W}$ satisfies the RIP$_1$.
\end{proof}

\subsubsection{Concentration Property of Sample Quantile}\label{s2.1.2}
\hskip\parindent

We proceed with the definition of the quantile of a population distribution and its sample counterpart.
\begin{definition}\label{quantiledef}
(Generalized quantile function). Let $0<\omega<1$. The generalized quantile function is defined as
\begin{equation}\label{GQF1}
F^{-1}(\omega)= \inf\{x\in \mathbb{R}: F(x)\geq \omega\},
\end{equation}
where $F(\cdot)$ is a cumulative distribution function. For simplicity, we denote $\theta_\omega(F)= F^{-1}(\omega)$ as the $\omega$-quantile of $F$. For a sample sequence $\{y_i\}_{i= 1}^m$, the sample $\omega$-quantile $\theta_\omega(\{y_i\}_{i= 1}^m)$ refers to $\theta_\omega(\hat{F})$, with $\hat{F}$ being the empirical distribution of the samples $\{y_i\}_{i= 1}^m$.
\end{definition}

We establish that as long as the sample size is sufficiently large, the sample quantile concentrates around the population quantile.
\begin{lemma}\label{Medianlem1}(\cite[Lemma 1]{Zhang2016PNPR})
Suppose $F(\cdot)$ is cumulative distribution function with a continuous probability density function $f(\cdot)$. If the samples $\{y_i\}_{i= 1}^m$ are i.i.d. drawn from $f$, and $0<\omega<1$, then the inequality $|\theta_\omega(\{y_i\}_{i= 1}^m)-\theta_\omega(F)|<\epsilon$
holds with a probability of at least $1-2\exp(-2m\epsilon^2l^2)$, provided that $l<f(\theta)<L$ for all $\theta$ in the set $\{\theta:|\theta-\theta_\omega|\leq \epsilon\}$.
\end{lemma}

We also recall a result related to outliers.
\begin{lemma}\label{outliersresult}(\cite[Lemma 3]{Zhang2016PNPR}, \cite[Lemma A.3]{Li2020MTGD})
Consider clean samples $\{\tilde{y}_i\}_{i= 1}^m$. If a fraction $p$ of these samples are corrupted by outliers, the resulting contaminated samples $\{y_i\}_{i= 1}^m$ contain $pm$ corrupted samples and $(1-p)m$ clean samples. For a quantile $\tau$ such that $p<\tau<1-p$, we have
\begin{equation*}\label{outliersineq}
\theta_{\tau-p}(\{\tilde{y}_i\}_{i= 1}^m)\leq\theta_{\tau}(\{y_i\}_{i= 1}^m)\leq \theta_{\tau+p}(\{\tilde{y}_i\}_{i= 1}^m).
\end{equation*}
\end{lemma}


\subsubsection{Some Necessary Inequalities}
\begin{lemma}\label{lem1.3}(\cite[Lemma 7.3]{Cai2015ROP})
Suppose that $\mathbf{z}\in \mathbb{R}^n$ is generalized with i.i.d. Gaussian entries and $z_i\sim \mathcal{N}(0, \sigma^2)$, we have $\mathbb{P}\left(\|\mathbf{z}\|_2\geq \sigma\sqrt{n+2\sqrt{n\log n}}\right)\leq 1/n$.
\end{lemma}

\begin{lemma}(\cite{Bourgain1987})\label{covernumber}
Let $\mathcal{D}_s= \{\mathbf{z}\in \mathbb{R}^n: \|\mathbf{z}\|_0\leq s, \|\mathbf{z}\|_2= 1\}$. Then there exists an $\iota$-net $\bar{\mathcal{D}}_s\subset \mathcal{D}_s$ with respect to the $\ell_2$-norm obeying $|\bar{\mathcal{D}}_s|\leq \left(3n/\iota\right)^s$.
\end{lemma}

\subsection{Proofs of Propositions \ref{prop1} and \ref{pro.ContractionOfInnerIteration}}
\hskip\parindent

Firstly, we give the proof of Proposition \ref{prop1} in details.
\begin{proof}[Proof of Proposition \ref{prop1}]
(i) First, we fix an arbitrary $(s+k)$-sparse signal $\tilde{\mathbf{x}}\in \mathbb{R}^n$ satisfying $\|\tilde{\mathbf{x}}\|_2=1$, and subsequently generalize to all $(s+k)$-sparse signals through a covering technique. It is noteworthy that the terms $m|\langle\mathbf{a}_i, \tilde{\mathbf{x}}\rangle|, \ i= 1, \cdots, m$, are i.i.d. replicates of $m|\mathbf{A}\tilde{\mathbf{x}}|$. Here, the matrix $\mathbf{A}$ is constructed with i.i.d. Gaussian entries $\mathcal{N}(0, \frac{1}{m^2})$. Given that $\|\tilde{\mathbf{x}}\|_2= 1$, the vector $m\mathbf{A}\tilde{\mathbf{x}}$ follows a standard normal distribution $\mathcal{N}(\mathbf{0}, \mathbf{I}_m)$, and thus $m|\mathbf{A}\tilde{\mathbf{x}}|$ follows a folded normal distribution. As established in Lemma \ref{Medianlem1}, the $\tau$-quantile $m\theta_{\tau}(|\mathbf{A}\tilde{\mathbf{x}}|)$ satisfies the inequality: 
\begin{equation}
\Phi^{-1}\left(\frac{1+\tau}{2}\right)-\epsilon\leq m\theta_{\tau}(|\mathbf{A}\tilde{\mathbf{x}}|)\leq \Phi^{-1}\left(\frac{1+\tau}{2}\right)+\epsilon
\end{equation}
with a probability of at least $1-2\exp(-dm\epsilon^2)$ for a small positive value $\epsilon$, where $d$ is a constant approximately equal to $2\times (2\phi(\frac{1+\tau}{2}))^2$, and $\phi$ represents the probability density function of the standard Gaussian distribution.

Next, we generalize the above result to all $(s+k)$-sparse signals $\mathbf{x}$ with $\|\mathbf{x}\|_2= 1$ by leveraging a covering argument. Let $\mathcal{N}_{\iota}$ stand for an $\iota$-net that covers the set of all $(s+k)$-sparse signals under the $\ell_2$-norm. According to Lemma \ref{covernumber}, the cardinality of $\mathcal{N}_{\iota}$ satisfies $|\mathcal{N}_{\iota}|\leq (\frac{3n}{\iota})^{s+k}$. Employing the union bound, we can derive that
\begin{equation}
 \Phi^{-1}\left(\frac{1+\tau}{2}\right)-\epsilon\leq m\theta_{\tau}(|\mathbf{A}\tilde{\mathbf{x}}|)\leq \Phi^{-1}\left(\frac{1+\tau}{2}\right)+\epsilon, \ \forall \ \tilde{\mathbf{x}} \in \mathcal{N}_{\iota}
\end{equation}
holds with probability at least $1-2(\frac{3n}{\iota})^{s+k}\exp(-dm\epsilon^2)$. We set $\iota= \frac{\epsilon}{\sqrt{n+2\sqrt{n\log n}}}$. Under this event and Lemma \ref{lem1.3}, for any $(s+k)$-sparse signal $\mathbf{x}$ with $\|\mathbf{x}\|_2= 1$, there exists a signal $\tilde{\mathbf{x}} \in \mathcal{N}_{\iota}$ such that $\|\mathbf{x}-\tilde{\mathbf{x}}\|_2\leq \iota$. Furthermore, by Lemma 2 of \cite{Zhang2016PNPR}, the following inequality holds
\begin{eqnarray}\label{eq3.7}
  \left|m\theta_{\tau}(|\mathbf{A}\tilde{\mathbf{x}}|)-m\theta_{\tau}(|\mathbf{A}\mathbf{x}|)\right| &\leq& \max_i m\left||\langle\mathbf{a}_i, \tilde{\mathbf{x}}\rangle|-|\langle\mathbf{a}_i, \mathbf{x}\rangle|\right| 
  \leq \max_i m\left|\langle\mathbf{a}_i, \tilde{\mathbf{x}}\rangle-\langle\mathbf{a}_i, \mathbf{x}\rangle\right| \nonumber \\
   &\leq& \max_i m\|\mathbf{x}-\tilde{\mathbf{x}}\|_2\|\mathbf{a}_i\|_2  
   \leq \iota\max_i \|m\mathbf{a}_i\|_2 \leq \epsilon
\end{eqnarray}
with a probability of at least $1-2(\frac{3n}{\iota})^{s+k}\exp(-dm\epsilon^2)-\frac{1}{n}$. 

The remaining part of the proof is to argue that \eqref{eq3.7} holds with probability at least $1-d_1\exp(-d_2m\epsilon^2)-\frac{1}{n}$, where $d_1$ and $d_2$ are some constants. This is valid provided that $m\geq d_0(\epsilon^{-2}\log (\epsilon^{-1}))(s+k)\log n$ for a sufficiently large constant $d_0$. To verify this, we first analyze the term $\left(\frac{3n}{\iota}\right)^{s+k}$:
\begin{align*}
  \left(\frac{3n}{\iota}\right)^{s+k} & = \exp\left((s+k)\log\left(\frac{3n}{\iota}\right)\right) \\
   & = \exp\left((s+k)\left(\log 3+ \log n+\frac{1}{2}\log(n+2\sqrt{n\log n})+\log(\epsilon^{-1})\right)\right) \\
   & \leq \exp\left((s+k)\left(\frac{3}{2}\log 3+\frac{3}{2}\log n+\log(\epsilon^{-1})\right)\right) \\
   & \leq \left(3(s+k)\log n+(s+k)\log(\epsilon^{-1})\right),
\end{align*}
where the first inequality is derived from the fact that $\log n< n$. It is easy to conform that $3(s+k)\log n<d_3m\epsilon^2$ and $(s+k)\log(\epsilon^{-1})\leq d_4m\epsilon^2$. Based on the specific formulation of $m$, as long as the constant $d_0$ is chosen to be sufficiently large, we can obtain 
\begin{equation*}
2\left(\frac{3n}{\iota}\right)^{s+k}\exp(-dm\epsilon^2)<d_1\exp(-d_2m\epsilon^2).
\end{equation*}
Here, $d_3+d_4<d-d_2$.

(ii) We first address the upper bound for $\|(\mathbf{b-Ax})\odot (\mathbb{I}_{\{|b_i-(\mathbf{Ax})_i|\leq \theta_\tau(|\mathbf{b-Ax}|)\}})_{i=1}^m\|_1$:
\begin{align*}
\|(\mathbf{b-Ax})\odot (\mathbb{I}_{\{|b_i-(\mathbf{Ax})_i|\leq \theta_\tau(|\mathbf{b-Ax}|)\}})_{i=1}^m\|_1&\leq \tau m\theta_\tau(|\mathbf{b-Ax}|)\\
&\leq \tau m\theta_{\tau+p}(|\mathbf{A}(\mathbf{x}_0-\mathbf{x})|)\\
&\leq \tau \left(\Phi^{-1}\left(\frac{1+\tau+p}{2}\right)+\epsilon\right)\|\mathbf{x-x}_0\|_2,
\end{align*}
where the second inequality is derived from Lemma \ref{outliersresult}, and the last inequality holds with high probability for a small $\epsilon$ from Item (i) above. Next, we establish the lower bound for $\|(\mathbf{b-Ax})\odot (\mathbb{I}_{\{|b_i-(\mathbf{Ax})_i|\leq \theta_\tau(|\mathbf{b-Ax}|)\}})_{i=1}^m\|_1$:
\begin{align*}
\|(\mathbf{b-Ax})\odot (\mathbb{I}_{\{|b_i-(\mathbf{Ax})_i|\leq \theta_\tau(|\mathbf{b-Ax}|)\}})_{i=1}^m\|_1&= \|(\mathbf{A}_{T_1}(\mathbf{x}_0-\mathbf{x}))+\boldsymbol{\eta}_{T_1}\|_1+\|\mathbf{A}_{T_2}(\mathbf{x}_0-\mathbf{x})\|_1\\
&\geq \|\mathbf{A}_{T_2}(\mathbf{x}_0-\mathbf{x})\|_1\\
&\geq \left(\tau-\frac{|T_1|}{m}\right)\sqrt{\frac{2}{\pi}}(1-\delta_{s+l})\|\mathbf{x-x}_0\|_2,
\end{align*}
where $T_1=T\cap \Gamma, T_2= \Gamma\setminus T_1$ with $|T_1\cup T_2|= \tau m$ and $|T_1|$ is a small number. Thus the proof is finished.
\end{proof}

Next, we turn our attention to Proposition \ref{pro.ContractionOfInnerIteration}.
\begin{proof}[Proof of Proposition \ref{pro.ContractionOfInnerIteration}]
(i) We first prove Item (i). For $\|[\mathbf{x}_0-\mathbf{u}^{k,l}-t_{k,l}\mathbf{A}^\top \mathrm{sign}(\mathbf{b}-\mathbf{Au}^{k,l})]_{\Lambda^{k}}\|_2^2$, it can be equivalently expressed as:
\begin{align}\label{Decomposition.ContractionOfInnerIteration}
&\|[\mathbf{x}_0-\mathbf{u}^{k,l}-t_{k,l}\mathbf{A}^\top \mathrm{sign}(\mathbf{b}-\mathbf{Au}^{k,l})]_{\Lambda^{k}}\|_2^2 \nonumber \\
&= \|\mathbf{x}_0-\mathbf{u}^{k,l}\|_2^2+t_{k,l}^2\|[\mathbf{A}^\top \mathrm{sign}(\mathbf{b}-\mathbf{Au}^{k,l})]_{\Lambda^{k}}\|_2^2  \nonumber \\
 &\quad -2t_{k,l}\langle\mathbf{x}_0-\mathbf{u}^{k,l}, [\mathbf{A}^\top \mathrm{sign}(\mathbf{b}-\mathbf{Au}^{k,l})]_{\Lambda^{k}}\rangle \nonumber \\
 &=: \|\mathbf{x}_0-\mathbf{u}^{k,l}\|_2^2+t_{k,l}^2F_1-2t_{k,l}F_2.
 \end{align}
Next, we estimate the upper bound of $F_1$ and the lower bound of $F_2$. Using the RIP$_1$, we can bound $F_1$ as follows:
\begin{align*}
F_1 &= \langle\mathbf{A}^\top \mathrm{sign}(\mathbf{b}-\mathbf{Au}^{k,l}), [\mathbf{A}^\top \mathrm{sign}(\mathbf{b}-\mathbf{Au}^{k,l})]_{\Lambda^{k}}\rangle \nonumber \\
    &= \langle\mathrm{sign}(\mathbf{b}-\mathbf{Au}^{k,l}), \mathbf{A}[\mathbf{A}^\top \mathrm{sign}(\mathbf{b}-\mathbf{Au}^{k,l})]_{\Lambda^{k}}\rangle \nonumber \\
    &\leq \|\mathbf{A}[\mathbf{A}^\top \mathrm{sign}(\mathbf{b}-\mathbf{Au}^{k,l})]_{\Lambda^{k}}\|_1 \nonumber \\
    &\leq \sqrt{\frac{2}{\pi}}(1+\delta_{2k+s-1})\|[\mathbf{A}^\top \mathrm{sign}(\mathbf{b}-\mathbf{Au}^{k,l})]_{\Lambda^{k}}\|_2,
\end{align*}
thus we can get the estimation of $F_1$ as follows
\begin{equation}\label{F1}
F_1\leq \frac{2}{\pi}(1+\delta_{2k+s-1})^2.
\end{equation}
For the lower bound of $F_2$, we have
\begin{align}\label{F2}
F_2 &= \langle\mathbf{x}_0-\mathbf{u}^{k,l}, \mathbf{A}^\top \mathrm{sign}(\mathbf{b}-\mathbf{Au}^{k,l})\rangle= \langle\mathbf{A}(\mathbf{x}_0-\mathbf{u}^{k,l}), \mathrm{sign}(\mathbf{b}-\mathbf{Au}^{k,l})\rangle \nonumber \\
    &= \langle\mathbf{A}_T(\mathbf{x}_0-\mathbf{u}^{k,l}), \mathrm{sign}(\mathbf{A}_T(\mathbf{x}_0-\mathbf{u}^{k,l})+\boldsymbol{\eta}_T)\rangle \nonumber \\
    &\quad + \langle\mathbf{A}_{T^c}(\mathbf{x}_0-\mathbf{u}^{k,l}), \mathrm{sign}(\mathbf{A}_{T^c}(\mathbf{x}_0-\mathbf{u}^{k,l}))\rangle \nonumber \\
    &= \langle\mathbf{A}_T(\mathbf{x}_0-\mathbf{u}^{k,l})+\boldsymbol{\eta}_T, \mathrm{sign}(\mathbf{A}_T(\mathbf{x}_0-\mathbf{u}^{k,l})+\boldsymbol{\eta}_T)\rangle \nonumber \\
    &\quad - \langle \boldsymbol{\eta}_T, \mathrm{sign}(\mathbf{A}_T(\mathbf{x}_0-\mathbf{u}^{k,l})+\boldsymbol{\eta}_T)\rangle+ \|\mathbf{A}_{T^c}(\mathbf{x}_0-\mathbf{u}^{k,l})\|_1 \nonumber \\
    &\geq \|\mathbf{A}_T(\mathbf{x}_0-\mathbf{u}^{k,l})+\boldsymbol{\eta}_T\|_1-\|\boldsymbol{\eta}_T\|_1+\|\mathbf{A}_{T^c}(\mathbf{x}_0-\mathbf{u}^{k,l})\|_1 \nonumber \\
    &\geq -\|\mathbf{A}_T(\mathbf{x}_0-\mathbf{u}^{k,l})\|_1+\|\mathbf{A}_{T^c}(\mathbf{x}_0-\mathbf{u}^{k,l})\|_1 \nonumber \\
    &= -\|\mathbf{A}(\mathbf{x}_0-\mathbf{u}^{k,l})\|_1+2\|\mathbf{A}_{T^c}(\mathbf{x}_0-\mathbf{u}^{k,l})\|_1 \nonumber \\
    &\geq \sqrt{\frac{2}{\pi}}[(2-2p)(1-\delta_{2k+s-1})-(1+\delta_{2k+s-1})]\|\mathbf{x}_0-\mathbf{u}^{k,l}\|_2 \nonumber \\
    &=: \sqrt{\frac{2}{\pi}}c_k\|\mathbf{x}_0-\mathbf{u}^{k,l}\|_2,
\end{align}
where the last inequality comes from RIP$_1$ and Lemma \ref{RIP1}, and $c_k= (2-2p)(1-\delta_{2k+s-1})-(1+\delta_{2k+s-1})$. Combining the bounds for $F_1$ and $F_2$ with \eqref{Decomposition.ContractionOfInnerIteration}, we arrive at
\begin{align*}
&\|[\mathbf{x}_0-\mathbf{u}^{k,l}-t_{k,l}\mathbf{A}^\top \mathrm{sign}(\mathbf{b}-\mathbf{Au}^{k,l})]_{\Lambda^{k}}\|_2^2 \nonumber \\
&\leq \|\mathbf{x}_0-\mathbf{u}^{k,l}\|_2^2+t_{k,l}^2\frac{2}{\pi}(1+\delta_{2k+s-1})^2- 2t_{k,l}\sqrt{\frac{2}{\pi}}c_k\|\mathbf{x}_0-\mathbf{u}^{k,l}\|_2 \nonumber \\
   &= \|\mathbf{x}_0-\mathbf{u}^{k,l}\|_2^2+\mu_{k,l}^2(1+\delta_{2k+s-1})^2\|(\mathbf{b-Au}^{k,l})\odot (\mathbb{I}_{\{|b_i-(\mathbf{Au}^{k,l})_i|\leq \theta_\tau(|\mathbf{b-Au}^{k,l}|)\}})_{i=1}^m\|_1^2 \nonumber \\
   &\quad - 2\mu_{k,l}c_k\|(\mathbf{b-Au}^{k,l})\odot (\mathbb{I}_{\{|b_i-(\mathbf{Au}^{k,l})_i|\leq \theta_\tau(|\mathbf{b-Au}^{k,l}|)\}})_{i=1}^m\|_1\|\mathbf{x}_0-\mathbf{u}^{k,l}\|_2 \nonumber \\
   &\leq \left(1+\tau^2(\Phi^{-1}+\epsilon)^2(1+\delta_{2k+s-1})^2\mu_{k,l}^2- 2c_k\sqrt{\frac{2}{\pi}}\left(\tau-\frac{|T_1^{k,l}|}{m}\right)(1-\delta_{2k+s-1})\mu_{k,l}\right)\|\mathbf{x}_0-\mathbf{u}^{k,l}\|_2^2 \nonumber \\
   &=: \rho_{k,l}\|\mathbf{x}_0-\mathbf{u}^{k,l}\|_2^2,
\end{align*}
where the last inequality is from Proposition \ref{prop1}. Thus we finish the proof of item (i).

(ii) Next we turn our attention to Item (ii). 
 Firstly, we proceed with the analysis by defining the supports of the vectors involved in the iterative process. Let $S, S^{k-1}$ and $S^{k} (k\geq s)$ denote the supports of $\mathbf{x}_0, \mathbf{x}^{k-1}$ (or $\mathbf{u}^{k,0}$), and $\mathbf{u}^{k,l}$ for $1\leq l\leq L+1$, respectively, and let $\Lambda^{k}:= S\cup S^{k-1}\cup S^{k}$.

It follows from the update scheme \eqref{GFHTP} that
\begin{align}\label{errorestimation.eq1}
\|\mathbf{x}_0-\mathbf{u}^{k,l+1}\|_2^2&=\|(\mathbf{x}_0-\mathbf{u}^{k,l+1})_{S^k}\|_2^2+\|(\mathbf{x}_0-\mathbf{u}^{k,l+1})_{(S^k)^c}\|_2^2 \nonumber \\
 &= \|[\mathbf{x}_0-\mathbf{u}^{k,l}-t_{k,l}\mathbf{A}^\top \mathrm{sign}(\mathbf{b}-\mathbf{Au}^{k,l})]_{S^{k}}\|_2^2+\|(\mathbf{x}_0)_{S\backslash S^k}\|_2^2 \nonumber \\
 &\leq \|[\mathbf{x}_0-\mathbf{u}^{k,l}-t_{k,l}\mathbf{A}^\top \mathrm{sign}(\mathbf{b}-\mathbf{Au}^{k,l})]_{\Lambda^{k}}\|_2^2+\|(\mathbf{x}_0)_{S\backslash S^k}\|_2^2 \nonumber \\
 &=: E_1+E_2.
\end{align}
To achieve this, we consider the upper bounds of the terms $E_1$ and $E_2$.

Note that it follows from Item (i) that 
\begin{equation}\label{errorestimation.E1}
E_1\leq  \rho_{k,l}\|\mathbf{x}_0-\mathbf{u}^{k,l}\|_2^2.
\end{equation}
Therefore, it remains to estimate $E_2$.
We claim that  $E_2$ satisfies
\begin{equation}\label{errorestimation.E2}
  E_2 \leq 2\rho_{k,0}\|\mathbf{x}_0-\mathbf{x}^{k-1}\|_2^2.
\end{equation}
By substituting the two bounds above into \eqref{errorestimation.eq1}, one has
\begin{align}\label{errorestimation.eq3.add1}
\|\mathbf{x}_0-\mathbf{u}^{k,l+1}\|_2^2\leq \rho_{k,l}\|\mathbf{x}_0-\mathbf{u}^{k,l}\|_2^2+2\rho_{k,0}\|\mathbf{x}_0-\mathbf{x}^{k-1}\|_2^2.
\end{align}

To conclude, proving the inequality  \eqref{errorestimation.E2} is sufficient to establish the conclusion.
Since $S^k$ is the index set of $k$ largest absolute entries of $\mathbf{x}^{k-1}+t_{k,0}\mathbf{A}^\top\mathrm{sign}(\mathbf{b}-\mathbf{Ax}^{k-1})$, then we notice that
\begin{equation*}
  \|(\mathbf{x}^{k-1}+t_{k,0}\mathbf{A}^\top\mathrm{sign}(\mathbf{b}-\mathbf{Ax}^{k-1}))_{S^k}\|_2\geq \|(\mathbf{x}^{k-1}+t_{k,0}\mathbf{A}^\top\mathrm{sign}(\mathbf{b}-\mathbf{Ax}^{k-1}))_{S}\|_2.
\end{equation*}
By eliminating the contribution on $S\cap S^k$, we obtain
\begin{equation}\label{errorestimation.eq2}
  \|(\mathbf{x}^{k-1}+t_{k,0}\mathbf{A}^\top\mathrm{sign}(\mathbf{b}-\mathbf{Ax}^{k-1}))_{S^k\backslash S}\|_2\geq \|(\mathbf{x}^{k-1}+t_{k,0}\mathbf{A}^\top\mathrm{sign}(\mathbf{b}-\mathbf{Ax}^{k-1}))_{S\backslash S^k}\|_2.
\end{equation}
The left-hand side satisfies
\begin{equation*}
  \|(\mathbf{x}^{k-1}+t_{k,0}\mathbf{A}^\top\mathrm{sign}(\mathbf{b}-\mathbf{Ax}^{k-1}))_{S^k\backslash S}\|_2= \|(\mathbf{x}^{k-1}-\mathbf{x}_0+t_{k,0}\mathbf{A}^\top\mathrm{sign}(\mathbf{b}-\mathbf{Ax}^{k-1}))_{S^k\backslash S}\|_2,
\end{equation*}
while the right-hand side satisfies
\begin{equation*}
  \|(\mathbf{x}^{k-1}+t_{k,0}\mathbf{A}^\top\mathrm{sign}(\mathbf{b}-\mathbf{Ax}^{k-1}))_{S\backslash S^k}\|_2\geq \|(\mathbf{x}_0)_{S\backslash S^k}\|_2-\|(\mathbf{x}^{k-1}-\mathbf{x}_0+t_{k,0}\mathbf{A}^\top\mathrm{sign}(\mathbf{b}-\mathbf{Ax}^{k-1}))_{S\backslash S^k}\|_2.
\end{equation*}
Consequently, substituting the two estimations above into the inequality \eqref{errorestimation.eq2},  we get
\begin{align*}
  E_2^{\frac{1}{2}}= \|(\mathbf{x}_0)_{S\backslash S^k}\|_2 &\leq \|(\mathbf{x}^{k-1}-\mathbf{x}_0+t_{k,0}\mathbf{A}^\top\mathrm{sign}(\mathbf{b}-\mathbf{Ax}^{k-1}))_{S\backslash S^k}\|_2 \\
  & \ \ \ +\|(\mathbf{x}^{k-1}-\mathbf{x}_0+t_{k,0}\mathbf{A}^\top\mathrm{sign}(\mathbf{b}-\mathbf{Ax}^{k-1}))_{S^k\backslash S}\|_2  \\
  & \leq \sqrt{2}\|(\mathbf{x}^{k-1}-\mathbf{x}_0+t_{k,0}\mathbf{A}^\top\mathrm{sign}(\mathbf{b}-\mathbf{Ax}^{k-1}))_{S\bigtriangleup S^k}\|_2 \\
  & \leq \sqrt{2}\|(\mathbf{x}^{k-1}-\mathbf{x}_0+t_{k,0}\mathbf{A}^\top\mathrm{sign}(\mathbf{b}-\mathbf{Ax}^{k-1}))_{\Lambda_k}\|_2 \\
  & \leq \sqrt{2\rho_{k,0}}\|\mathbf{x}_0-\mathbf{x}^{k-1}\|_2,
\end{align*}
where the last inequality is from $E_{1}=\|[\mathbf{x}_0-\mathbf{u}^{k,l}-t_{k,l}\mathbf{A}^\top \mathrm{sign}(\mathbf{b}-\mathbf{Au}^{k,l})]_{\Lambda^{k}}\|_2^2$ and Proposition \ref{pro.ContractionOfInnerIteration} with $l= 0$. Thus, we get the estimation for $E_2$. Thus we finish the conclusion. 

\end{proof}

\subsection{Proof of Theorem \ref{thm3.1}}\label{s2.2}
\hskip\parindent

In this subsection, we show the proof of Theorem \ref{thm3.1} in details.
\begin{proof}[Proof of Theorem \ref{thm3.1}]
Notice that it follows from Proposition \ref{pro.ContractionOfInnerIteration} that 
\begin{align}\label{errorestimation.eq3.add1}
\|\mathbf{x}_0-\mathbf{u}^{k,l+1}\|_2^2\leq \rho_{k,l}\|\mathbf{x}_0-\mathbf{u}^{k,l}\|_2^2+2\rho_{k,0}\|\mathbf{x}_0-\mathbf{u}^{k,0}\|_2^2.
\end{align}
Then by conducting immediate induction on $l$ with $\mathbf{u}^{k,0}=: \mathbf{x}^{k-1}$ and $\mathbf{u}^{k,L+1}=: \mathbf{x}^{k}$, we have
\begin{align}\label{errorestimation.eq3}
\|\mathbf{x}_0-\mathbf{x}^{k}\|_2^2&= \|\mathbf{x}_0-\mathbf{u}^{k,L+1}\|_2^2\nonumber\\
&\leq \rho_{k,L}\|\mathbf{x}_0-\mathbf{u}^{k,L}\|_2^2+2\rho_{k,0}\|\mathbf{x}_0-\mathbf{x}^{k-1}\|_2^2\nonumber\\
&\leq\rho_{k,L}\left(\rho_{k,L-1}\|\mathbf{x}_0-\mathbf{u}^{k,L-1}\|_2^2+2\rho_{k,0}\|\mathbf{x}_0-\mathbf{x}^{k-1}\|_2^2\right)+2\rho_{k,0}\|\mathbf{x}_0-\mathbf{x}^{k-1}\|_2^2\nonumber\\
&= \left(\prod_{l=L-1}^L\rho_{k,l}\right)\|\mathbf{x}_0-\mathbf{u}^{k,L-1}\|_2^2+2\rho_{k,0}(1+\rho_{k,L})\|\mathbf{x}_0-\mathbf{x}^{k-1}\|_2^2\nonumber\\
&\leq\ldots\ldots\nonumber\\
&\leq \left(\prod_{l=0}^L\rho_{k,l}\right)\|\mathbf{x}_0-\mathbf{u}^{k,0}\|_2^2+2\rho_{k,0}\left(1+\sum_{i=1}^{L}\prod_{l=i}^L\rho_{k,l}\right)\|\mathbf{x}_0-\mathbf{x}^{k-1}\|_2^2.
\end{align}
In summary, due to $\rho_{k}= \max_{l}\rho_{k,l}$, we obtain
\begin{align*}
\|\mathbf{x}_0-\mathbf{x}^{k}\|_2^2 
&\leq \left(\prod_{l=0}^L\rho_{k}\right)\|\mathbf{x}_0-\mathbf{x}^{k-1}\|_2^2+2\rho_{k}\left(1+\sum_{i=1}^{L}\prod_{l=i}^L\rho_{k}\right)\|\mathbf{x}_0-\mathbf{x}^{k-1}\|_2^2\nonumber\\
&=\left(\rho_{k}^{L+1}+2\rho_{k}(\rho_{k}^{L}+\rho_{k}^{L-1}+\ldots+\rho_{k}+1)\right)\|\mathbf{x}_0-\mathbf{x}^{k-1}\|_2^2\nonumber\\
&=\left(\rho_{k}^{L+1}+2\rho_{k}\frac{1-\rho_{k}^{L+1}}{1-\rho_{k}}\right)\|\mathbf{x}_0-\mathbf{x}^{k-1}\|_2^2\nonumber\\
&= \left(\frac{\rho_{k}^{L+1}(1-3\rho_{k})+2\rho_{k}}{1-\rho_{k}}\right)\|\mathbf{x}_0-\mathbf{x}^{k-1}\|_2^2.
\end{align*}
Here, the condition \eqref{GFHTP_condition} implies that $\frac{\rho_{k}^{L+1}(1-3\rho_{k})+2\rho_{k}}{1-\rho_{k}}<1$. 
\end{proof}

\section{Proof for Structured Sparse Signals}\label{secappendixb}
\subsection{Proof of Proposition \ref{prop.IndificationSupportSet} }
\begin{proof}
Firstly, we claim that 
\begin{align}\label{StructureSignal.eq3}
 \|\mathbf{x}^{k-1}-\mathbf{x}_0\|_2^2 
  & \leq 2(\beta_{k}+\tau^2(\Phi^{-1}+\epsilon)^2(1+\delta_{s})^2\mu_{k,0}^2)\|\mathbf{x}_0-\mathbf{x}^{k-1}\|_2^2+\|(\mathbf{x}_0)_{(S^{k-1})^c}\|_2^2 \nonumber\\
  & =: 2\beta'_{k}\|\mathbf{x}^{k-1}-\mathbf{x}_0\|_2^2+\|(\mathbf{x}_0)_{(S^{k-1})^c}\|_2^2,
\end{align}
where $\beta_k= 1+\tau^2(\Phi^{-1}+\epsilon)^2(1+\delta_{s})^2\mu_{k,0}^2- 2c\sqrt{\frac{2}{\pi}}\left(\tau-\frac{|T_1^{k,0}|}{m}\right)(1-\delta_{s})\mu_{k,0}$. And the upper bound of $\xi_k$ satisfies
\begin{align}\label{Estimation.xi}
  \xi_k = & t_{k,0}\max_{\ell\in S^c}|\mathbf{A}^\top\mathrm{sign}(\mathbf{b-Ax}^{k-1}))_{\ell}| \nonumber\\
  \leq & \mu_{k,0}\sqrt{\frac{\pi}{2}}\tau(\Phi^{-1}+\epsilon)\|\mathbf{x}^{k-1}-\mathbf{x}_0\|_2\max_{\ell\in S^c}|\mathbf{A}^\top\mathrm{sign}(\mathbf{b-Ax}^{k-1}))_{\ell}|,
\end{align}
and  the lower bound of $\zeta_k$ is 
\begin{align}\label{Estimation.zeta}
  \zeta_k
  &\geq \frac{1}{\sqrt{s-k+1}}\left(\frac{\sqrt{1-2\beta'_{k}}}{\lambda}-\sqrt{\beta_{k}}\right)\|\mathbf{x}^{k-1}-\mathbf{x}_0\|_2.
\end{align}

By the combination of estimations \eqref{Estimation.xi} and \eqref{Estimation.zeta}, we have
\begin{align*}
  & \mathbb{P}(\xi_k\geq \zeta_k) \\
  & \leq \mathbb{P}\left(\mu_{k,0}\sqrt{\frac{\pi}{2}}\tau(\Phi^{-1}+\epsilon)\max_{\ell\in S^c}|(\mathbf{A}^\top\mathrm{sign}(\mathbf{b-Ax}^{k-1}))_{\ell}|\geq \frac{1}{\sqrt{s-k+1}}\left(\frac{\sqrt{1-2\beta'_{k}}}{\lambda}-\sqrt{\beta_{k}}\right)\right)\\
  & = \mathbb{P}\left(\max_{\ell\in S^c}|\langle\tilde{\mathbf{a}}_{\ell},\mathrm{sign}(\mathbf{b-Ax}^{k-1})\rangle|\geq \frac{\sqrt{\frac{2}{\pi}}}{\sqrt{s-k+1}\mu_{k,0}\tau(\Phi^{-1}+\epsilon)}\left(\frac{\sqrt{1-2\beta'_{k}}}{\lambda}-\sqrt{\beta_{k}}\right)\right)\\
  &\leq \mathbb{P}\left(\max_{\ell\in S^c}|\langle\tilde{\mathbf{a}}_{\ell},\mathrm{sign}(\mathbf{b-Ax}^{k-1})\rangle|\geq \frac{\sqrt{\frac{2}{\pi}}}{\sqrt{s}\mu_{k,0}\tau(\Phi^{-1}+\epsilon)}\left(\frac{\sqrt{1-2\beta'_{k}}}{\lambda}-\sqrt{\beta_{k}}\right)\right)\\
  & =: \mathbb{P}\left(\max_{\ell\in S^c}|\langle\tilde{\mathbf{a}}_{\ell},\mathrm{sign}(\mathbf{b-Ax}^{k-1})\rangle|\geq\frac{\gamma_{k}}{\sqrt{s}}\right),
\end{align*}
where $\mathbf{A}=[\tilde{\mathbf{a}}_1, \cdots, \tilde{\mathbf{a}}_n]$.
Note that for the sum of these random variables $S_m=X_1+\ldots+X_m$ with independent and bounded random variables 
 $X_j\in[a_j,b_j]$ for $j=1,\cdots,m$, the Hoeffding's inequality \cite{Hoeffding1963} states that for all $t>0$, $\mathbb{P}(|S_m-\mathbb{E}(S_m)|\geq t)\leq 2\exp(-2t^2/\sum_{j=1}^{m}(b_j-a_j)^2)$. Thus, from $\mathbb{E}\langle\tilde{\mathbf{a}}_{\ell},\mathrm{sign}(\mathbf{b-Ax}^{k-1})\rangle=0$, we can conclude that
\begin{align*}
  \mathbb{P}(\xi_k\geq \zeta_k) & \leq \sum_{\ell\in S^c}\mathbb{P}\left(|\langle\tilde{\mathbf{a}}_{\ell},\mathrm{sign}(\mathbf{b-Ax}^{k-1})\rangle|\geq
  \frac{\gamma_{k}}{\sqrt{s}}\right) \\
  &\leq 2(n-s)\exp\left(-\frac{\gamma_{k}^2}{2s(\max_{\ell}\|\tilde{\mathbf{a}}_{\ell}\|_2^2)}\right) \\
  &\leq 2(n-s)\exp\left(-\frac{\gamma_{k}^2}{2s\left(\frac{1}{m}\sqrt{m+2\sqrt{m\log m}}\right)^2}\right)\\
  &\leq 2(n-s)\exp\left(-\frac{\gamma_{k}^2m}{6s}\right),
\end{align*}
where the third inequality is derived from Lemma \ref{lem1.3}.

To conclude, it is sufficient to prove the three estimations \eqref{StructureSignal.eq3}, \eqref{Estimation.xi} and \eqref{Estimation.zeta}.
\begin{itemize}
\item[(i)] Proof of the Inequality \eqref{StructureSignal.eq3}. 
Using the triangle inequality and the RIP$_1$ condition, we obtain
\begin{align}\label{StructureSignal.eq2}
  \|\mathbf{x}^{k-1}-\mathbf{x}_0\|_2^2 & = \|(\mathbf{x}^{k-1}-\mathbf{x}_0)_{S^{k-1}}\|_2^2+\|(\mathbf{x}^{k-1}-\mathbf{x}_0)_{(S^{k-1})^c}\|_2^2 \nonumber\\
  & \leq 2\|(\mathbf{x}^{k-1}-\mathbf{x}_0+t_{k,0}\mathbf{A}^\top\mathrm{sign}(\mathbf{b-Ax}^{k-1}))_{S^{k-1}}\|_2^2 \nonumber\\
  & \ \ + 2t_{k,0}^2\|(\mathbf{A}^\top\mathrm{sign}(\mathbf{b-Ax}^{k-1}))_{S^{k-1}}\|_2^2+\|(\mathbf{x}_0)_{(S^{k-1})^c}\|_2^2 \nonumber\\
  & \leq 2\|(\mathbf{x}^{k-1}-\mathbf{x}_0+t_{k,0}\mathbf{A}^\top\mathrm{sign}(\mathbf{b-Ax}^{k-1}))_{S}\|_2^2 \nonumber\\
  & \ \ + 2t_{k,0}^2\|(\mathbf{A}^\top\mathrm{sign}(\mathbf{b-Ax}^{k-1}))_{S}\|_2^2+\|(\mathbf{x}_0)_{(S^{k-1})^c}\|_2^2,
\end{align}
where the second inequality comes from $S^{k-1}\subseteq S$. Notice that by the estimation \eqref{F1} and  Proposition \ref{prop1}, one has
\begin{align}\label{StructureSignal.eq5}
&t_{k,0}^2\|(\mathbf{A}^\top\mathrm{sign}(\mathbf{b-Ax}^{k-1}))_{S}\|_2^2
=: t_{k,0}^2G_1\nonumber\\
&\leq \frac{2}{\pi}(1+\delta_{s})^2 \left(\mu_{k,0}\sqrt{\frac{\pi}{2}}\|(\mathbf{b-Ax}^{k-1})\odot (\mathbb{I}_{\{|b_i-(\mathbf{Ax}^{k-1})_i|\leq \theta_\tau(|\mathbf{b-Ax}^{k-1}|)\}})_{i=1}^m\|_1\right)^2
 \nonumber\\
&\leq \tau^2(\Phi^{-1}+\epsilon)^2(1+\delta_{s})^2\mu_{k,0}^2 \|\mathbf{x}_0-\mathbf{x}^{k-1}\|_2^2.
\end{align}

Given that $\zeta_{k-1}>\xi_{k-1}$ implies $S^{k-1}\subseteq S$, we can further get
\begin{align}\label{Decomposition}
 & \|(\mathbf{x}^{k-1}-\mathbf{x}_0+t_{k,0}\mathbf{A}^\top\mathrm{sign}(\mathbf{b-Ax}^{k-1}))_{S}\|_2^2 \nonumber\\
  = & \|\mathbf{x}^{k-1}-\mathbf{x}_0\|_2^2+t_{k,0}^2\|(\mathbf{A}^\top\mathrm{sign}(\mathbf{b-Ax}^{k-1}))_{S}\|_2^2-2t_{k,0}\langle\mathbf{x}_0-\mathbf{x}^{k-1}, (\mathbf{A}^\top\mathrm{sign}(\mathbf{b-Ax}^{k-1}))_{S}\rangle \nonumber\\
  =: & \|\mathbf{x}^{k-1}-\mathbf{x}_0\|_2^2+t_{k,0}^2G_1-2t_{k,0}G_2.
\end{align}
Consequently, leveraging a proof strategy analogous to that employed in Proposition \ref{pro.ContractionOfInnerIteration}, we can straightly derive the following desired result
\begin{align}\label{StructureSignal.eq1}
& \|(\mathbf{x}^{k-1}-\mathbf{x}_0+t_{k,0}\mathbf{A}^\top\mathrm{sign}(\mathbf{b-Ax}^{k-1}))_{S}|\|_2^2 \nonumber\\
\leq & \left(1+\tau^2(\Phi^{-1}+\epsilon)^2(1+\delta_{s})^2\mu_{k,0}^2- 2c\sqrt{\frac{2}{\pi}}\left(\tau-\frac{|T_1^{k,0}|}{m}\right)(1-\delta_{s})\mu_{k,0}\right)\|\mathbf{x}_0-\mathbf{x}^{k-1}\|_2^2 \nonumber\\
=: & \beta_{k}\|\mathbf{x}^{k-1}-\mathbf{x}_0\|_2^2,
\end{align}
where $c= (2-2p)(1-\delta_s)-(1+\delta_s)$. Thus by substituting \eqref{StructureSignal.eq5} and \eqref{StructureSignal.eq1} into \eqref{StructureSignal.eq2}, we get the estimation \eqref{StructureSignal.eq3}.

\item[(ii)] Proof of the Inequality \eqref{Estimation.xi}. It is clear that the estimation  \eqref{Estimation.xi} comes from Proposition \ref{prop1}.  

\item[(iii)] Proof of the Inequality \eqref{Estimation.zeta}. To prove \eqref{Estimation.zeta},  we first introduce an index set $T^{s-k+1}$, which corresponds to the $s-(k-1)$ smallest values of $|(\mathbf{x}^{k-1}+t_{k,0}\mathbf{A}^\top\mathrm{sign}(\mathbf{b-Ax}^{k-1}))_j|$ for $j$ in $S$. This allows us to derive a lower bound for $\zeta_k$: 
\begin{align*}
  \zeta_k\geq & \frac{1}{\sqrt{s-k+1}}\|(\mathbf{x}^{k-1}+t_{k,0}\mathbf{A}^\top\mathrm{sign}(\mathbf{b-Ax}^{k-1}))_{T^{s-k+1}}\|_2 \\
  \geq & \frac{1}{\sqrt{s-k+1}}\left(\|(\mathbf{x}_0)_{T^{s-k+1}}\|_2-\|(\mathbf{x}^{k-1}-\mathbf{x}_0+t_{k,0}\mathbf{A}^\top\mathrm{sign}(\mathbf{b-Ax}^{k-1}))_{T^{s-k+1}}\|_2\right).
\end{align*}
Since $\|(\mathbf{x}^{k-1}-\mathbf{x}_0+t_{k,0}\mathbf{A}^\top\mathrm{sign}(\mathbf{b-Ax}^{k-1}))_{T^{s-k+1}}\|_2\leq \|(\mathbf{x}^{k-1}-\mathbf{x}_0+t_{k,0}\mathbf{A}^\top\mathrm{sign}(\mathbf{b-Ax}^{k-1}))_{S}\|_2$ $(T^{s-k+1}\subseteq S)$ and the estimation \eqref{StructureSignal.eq1}, we derive that 
\begin{align*}
  \zeta_k &\geq \frac{1}{\sqrt{s-k+1}}\left(\|(\mathbf{x}_0)_{T^{s-k+1}}\|_2-\sqrt{\beta_{k}}\|\mathbf{x}^{k-1}-\mathbf{x}_0\|_2\right).
\end{align*}

Notice that \eqref{StructureSignal.eq3} gives a lower bound of $\|(\mathbf{x}_0)_{(S^{k-1})^c}\|_2$, which depends on $\|\mathbf{x}^{k-1}-\mathbf{x}_0\|_2$.
Given that $\beta'_{k}<1/2$, we can establish the relationship of $\|\mathbf{x}^{k-1}-\mathbf{x}_0\|_2$ and $\|(\mathbf{x}_0)_{T^{s-k+1}}\|_2$ as follows
\begin{align*}
  \|\mathbf{x}^{k-1}-\mathbf{x}_0\|_2 & \leq \frac{1}{\sqrt{1-2\beta'_{k}}}\|(\mathbf{x}_0)_{(S^{k-1})^c}\|_2 \leq\frac{\sqrt{s-k+1}}{\sqrt{1-2\beta'_{k}}}x_1^* \\
   & \leq \frac{\sqrt{s-k+1}}{\sqrt{1-2\beta'_{k}}}\lambda x_s^* \leq \frac{\lambda}{\sqrt{1-2\beta'_{k}}} \|(\mathbf{x}_0)_{T^{s-k+1}}\|_2
\end{align*}
with the final inequality stemming from the fact that $\sqrt{s-k+1}x_s^*\leq \|(\mathbf{x}_0)_{T^{s-k+1}}\|_2$ and $T^{s-k+1}\subset S$. This implies that 
\begin{equation}
\|(\mathbf{x}_0)_{T^{s-k+1}}\|_2\geq \frac{\sqrt{1-2\beta'_{k}}}{\lambda}\|\mathbf{x}^{k-1}-\mathbf{x}_0\|_2.
\end{equation}
This allows us to deduce the estimation \eqref{Estimation.zeta}. Thus, we complete the proof.
\end{itemize}

\end{proof}

\subsection{ Proof of Theorem \ref{propGFHTP}}
\hskip\parindent

To be end this section, we give the proof of Theorem \ref{propGFHTP}.
\begin{proof}[Proof of Theorem \ref{propGFHTP}]
We define the set $S$ as the support of $\mathbf{x}_0$.
Recall that $\zeta_k$ represents the $k$-th largest value of the elements in the subset $S$ of $|(\mathbf{x}^{k-1}+t_{k,0}\mathbf{A}^\top\mathrm{sign}(\mathbf{b-Ax}^{k-1}))_j|$, while $\xi_k$ denotes the largest value of the elements in the complementary subset $S^c$. Our objective is to establish that with high probability, $S^k\subseteq S$ for each $k\in [[s]]$. This is implied by the condition $\zeta_k>\xi_k$ for all $k\in [[s]]$. The probability of the event not holding is bounded by the following expression:
\begin{align*}
  P:= & \mathbb{P}(\exists k\in [[s]]: \xi_k\geq \zeta_k \ and \ (\zeta_{k-1}>\xi_{k-1}, \cdots, \zeta_1>\xi_1)) \nonumber\\
  \leq & \mathbb{P}\left(\left|\frac{\sqrt{\frac{\pi}{2}}\|\mathbf{Ax}\|_1}{\|\mathbf{x}\|_2}-1\right|>\delta_s \ for \ some \ s \ sparse \ \mathbf{x} \right)\nonumber \\
   &+ \sum_{k=1}^{s}\mathbb{P}\left(\xi_k\geq \zeta_k, (\zeta_{k-1}>\xi_{k-1}, \cdots, \zeta_1>\xi_1), \left(\left|\frac{\sqrt{\frac{\pi}{2}}\|\mathbf{Ax}\|_1}{\|\mathbf{x}\|_2}-1\right|\leq\delta_s \ for \ all \ s \ sparse \ \mathbf{x}\right)\right) \nonumber \\
   =:& P_1+P_2, \nonumber
\end{align*}
where the inequality comes from the probability formulation $\mathbb{P}(A)=\mathbb{P}(AB)+\mathbb{P}(A\bar{B})\leq \mathbb{P}(AB)+\mathbb{P}(\bar{B})$.
Referencing Lemma \ref{RIP}, we find that $P_1$ is constrained by $2\exp(-\frac{\delta_s^2}{16}m)$. 

Now we shift our attention to $P_2$ and employ the notation $\mathbb{P}(E)$ to denote the conditional probability of event $E$, given the intersection of two sets of conditions. The first set consists of the conditions
$(\zeta_{k-1}>\xi_{k-1}, \cdots, \zeta_1>\xi_1)$, and the second set involves the event that the RIP$_1$ holds for all $s$-sparse vectors $\mathbf{x}$. We assume that these events are valid for the purpose of this analysis. 

By Proposition \ref{prop.IndificationSupportSet},  we get the estimation for $P_2$ as follows
$$
P_2\leq 2s(n-s)\exp\left(-\frac{\gamma_{k-1}^2m}{6s}\right).
$$

Combining these results, the overall failure probability $P$ is bounded by
\begin{align*}
  P & \leq 2\exp\left(-\frac{\delta_s^2}{16}m\right)+2s(n-s)\exp\left(-\frac{\gamma_{k-1}^2m}{6s}\right).
\end{align*}
In particular, we observe that $\zeta_k>\xi_k$, which implies $S^s= S$, with a failure probability not exceeding $P$. Furthermore, since $|\frac{\sqrt{\frac{\pi}{2}}\|\mathbf{Ax}\|_1}{\|\mathbf{x}\|_2}-1|\leq\delta_s$ with a failure probability of $2\exp(-\frac{\delta_s^2}{16}m)$, we also find that
\begin{align*}
  \|\mathbf{x}_0-\mathbf{x}^s\|_2^2& = \|(\mathbf{x}^{s}-\mathbf{x}_0)_{S^s}\|_2^2+\|(\mathbf{x}^{s}-\mathbf{x}_0)_{(S^s)^c}\|_2^2 \nonumber\\
  &\leq 2\|(\mathbf{x}^{s}-\mathbf{x}_0+t_{s+1,0}\mathbf{A}^\top\mathrm{sign}(\mathbf{b-Ax}^{s}))_{S^{s}}\|_2^2 \\
  & \ \ + 2t_{s+1,0}^2\|\mathbf{A}^\top\mathrm{sign}(\mathbf{b-Ax}^{s}))_{S^{s}}\|_2^2+ \|(\mathbf{x}_0)_{S^c}\|_2^2\\
  &\leq 2\beta'_{s+1} \|\mathbf{x}_0-\mathbf{x}^s\|_2^2+ \|(\mathbf{x}_0)_{S^c}\|_2^2\\
  &<\|\mathbf{x}_0-\mathbf{x}^s\|_2^2,
\end{align*}
where the second inequality is derived from \eqref{StructureSignal.eq3} when $k=s+1$, and the last inequality is from $(\mathbf{x}_0)_{S^c}=\mathbf{0}$ and $\beta'_{s+1}<1/2$.
Thus, we obtain $\mathbf{x}^s= \mathbf{x}_0$. The resulting failure probability is constrained by
\begin{align*}
 & 4\exp\left(-\frac{\delta_s^2}{16}m\right)+2s(n-s)\exp\left(-\frac{\gamma_{k-1}^2m}{6s}\right)\\
 &\leq n\exp\left(-\frac{\delta_s^2}{16}m\right)+n^2\exp\left(-\frac{\gamma_{k-1}^2m}{6s}\right)\\
 &\leq n^2\exp\left(-c'\frac{m}{s}\right)\leq n^{-c''}.
\end{align*}
This bound holds with an appropriate selection of the constant $c'_1$ in the condition $m\geq c'_1s\log n$. Thus we finish the conclusion. 
\end{proof}

\bibliographystyle{plain}
\bibliography{GFHTP_reference}

\end{document}